\documentclass[pdflatex,sn-mathphys-num]{sn-jnl}

\usepackage{iftex}
\ifpdf
  \usepackage{underscore}         
  \usepackage[T1]{fontenc}        
\else
  \usepackage{breakurl}           
\fi

\usepackage{version} 
\includeversion{tarkabsout} 
\excludeversion{tarkabsin} 

\usepackage{graphicx}

\usepackage{amsmath,amsfonts, amsthm}
\usepackage{amssymb,hyperref}

\usepackage{color} 
\newcommand{\roncomment}[1]{}

\newcommand{\commentout}[1]{}

\newcommand{\kbp}{\mathbf{P}}

\newcommand{\decides}{\mathtt{decides}}
\newcommand{\decide}{\mathtt{decide}}

\newcommand{\Agents}{\mathtt{Agt}}

\newcommand{\exchange}{\mathcal{E}} 
\newcommand{\rimp}{\Rightarrow}
\newcommand{\dimp}{\Leftrightarrow}

\newcommand{\Time}{\mathit{time}}

\newcommand{\Nat}{\mathbb{N}}

\newcommand{\Prop}{Prop}

\newcommand{\I}{\mathcal{I}}

\newcommand{\R}{{\cal R}}

\newcommand{\failures}{\mathcal{F}} 
 
\newcommand{\crashed}{\mathit{crashed}}
\newcommand{\beln}{B^N} 

\newcommand{\bel}[1]{B^{#1}} 
\newcommand{\eb}[1]{E\hspace{-2pt}B_{#1}} 
\newcommand{\ek}[1]{E\hspace{-2pt}K_{#1}} 
\newcommand{\cb}[1]{C\hspace{-2pt}B_{#1}} 
\newcommand{\ck}[1]{C\hspace{-2pt}K_{#1}} 

\newcommand{\cbn}{C\hspace{-2pt}B_N} 

\newcommand{\Values}{V}
\newcommand{\noop}{\mathtt{noop}}

\newcommand{\adv}{\mathit{Adv}}

\newcommand{\N}{\mathcal{N}}
\newcommand{\A}{\mathcal{A}}

\newcommand{\init}{\mathit{init}}
\newcommand{\Msg}{M}
\newcommand{\Crash}{\mathit{Crash}}
\newcommand{\ComCrash}{\mathit{ComCrash}}

\newcommand{\known}{w} 
\newcommand{\new}{\mathit{new}} 
\newcommand{\done}{\mathit{done}} 
\newcommand{\kfaulty}{\mathit{kf}}

\theoremstyle{thmstyleone}%
\newtheorem{thm}{Theorem}
\newtheorem{theorem}[thm]{Theorem}
\newtheorem{proposition}[thm]{Proposition}
\newtheorem{lemma}[thm]{Lemma}
\newtheorem{corollary}[thm]{Corollary}

\theoremstyle{thmstyletwo}%

\theoremstyle{thmstylethree}%

\begin{document}

\title[Optimal Simultaneous Byzantine Agreement, Common Knowledge and Limited Information Exchange]{Optimal Simultaneous Byzantine Agreement, Common Knowledge and Limited Information Exchange}


\author*[1]{\fnm{Ron} \spfx{van der} \sur{Meyden}}\email{R.VanderMeyden@unsw.edu.au}

\affil*[1]{\orgdiv{School of Computer Science and Engineering}, \orgname{UNSW Sydney}, \orgaddress{\city{Sydney}, \postcode{2052}, \state{NSW}, \country{Australia}}}

\abstract{In order to develop solutions that perform actions as early
  as possible, analysis of distributed algorithms using epistemic
  logic has generally concentrated on ``full information protocols'',
  which may be inefficient with respect to space and computation time.
  The paper reconsiders the epistemic analysis of the problem of
  Simultaneous Byzantine Agreement with respect to weaker, but more
  practical, exchanges of information. The paper first clarifies some
  issues concerning both the specification of this problem and the
  knowledge based program characterizing its solution, concerning the
  distinction between the notions of ``nonfaulty'' and ``not yet
  failed'', on which there are variances in the literature.  It is
  then shown that, when implemented relative to a given failure model
  and an information exchange protocol satisfying certain conditions,
  this knowledge based program yields a protocol that is optimal
  relative to solutions using the same information exchange.
  Conditions are also identified under which this implementation is
  also an optimum, but an example is provided that shows this does not
  hold in general.}




\maketitle



\section{Introduction} 

The logic of knowledge has been shown to be a helpful formalism for the analysis of fault-tolerant distributed algorithms \cite{DM90,FHMVbook,MT88,HalpernMW01}. 
A particular focus of work in this area has been the problem of Byzantine Agreement \cite{PSL}, which requires a group of agents 
to coordinate on a decision in the face of 
faulty behaviour by some the agents. It has been shown that the precise conditions 
under which a decision can be made by an agent in such a setting 
can be characterized, independently of details of the fault model,  in terms of what the agent knows. That characterization can
then be applied to derive protocols that are \emph{optimal} in the sense that agents decide in each possible run, at the earliest possible time. 
The present paper reconsiders a number of issues in these results, for 
Simultaneous Byzantine Agreement (SBA), which requires agents to
decide simultaneously (in the same round of computation). This version of Byzantine Agreement is relevant for applications 
such as the fair release of stock market information, or the coordination of a set of 
actuators controlling physical equipment such as an airplane or motor vehicle.

In order to coordinate, agents need to exchange information. In the context of Byzantine Agreement protocols, 
this information is about the agents' initial preferences for the joint decision  to be made, and about the 
faults that they have observed while running the protocol. 
Driven by a focus on 
protocols that are theoretically optimal, in the sense of deciding as early as possible, 
the literature has concentrated on ``full information protocols'' \cite{PSL,DM90,MT88}, which maximize 
the information exchanged by having agents store all messages that they receive,
and transmitting their complete state in each round of the protocol. 
 Agents using a full information protocol know everything that they could know 
in any other protocol, enabling them to make their decision at a time no later than 
they would in any other protocol. 

However, full information protocols use agent states that grow exponentially with time. 
While this state can be  reduced with further analysis \cite{MT88},  in some cases, the theoretically 
optimal protocols are relatively inefficient, or even intractable, in space usage or computation time \cite{MT88,Moses09}.
Full information protocols are therefore not necessarily practical, 
and more practical protocols need to make compromises. 

Limiting the information exchanged 
by the protocol is one approach to obtaining a more practical protocol. 
However, one might still ask for a protocol that is optimal, when compared with other protocols exchanging information in the same way. 
Consideration of this issue was begun by Alpturer et al.~\cite{AHM23}, for the Eventual Byzantine Agreement problem in the case of sending omission failures. 
In the present paper, we consider optimality of limited information exchange protocols for 
SBA.
Our particular focus is to understand the relationship between optimality of SBA protocols relative to a limited information exchange 
and a knowledge based program for this problem. We are interested in a general result that covers a range of different 
failure and information exchange models, since this kind of abstraction is one of the advantages obtainable from 
 the knowledge based approach to distributed computing. 

In addressing this question, we 
first revisit a number of issues. 
The characterization of SBA protocols using the logic of knowledge 
has employed a number of distinct notions of common knowledge, 
and there are also differences in the underlying semantic models used to represent the various failure models 
that have been studied. It also emerges that there are subtleties with respect to the notion of optimality 
guaranteed by the knowledge based program once one considers limited information exchange. 

With respect to notions of common knowledge, the original analysis of 
SBA
in the  crash failures model by 
Dwork and Moses \cite{DM90} uses a notion of common knowledge amongst the \emph{nonfailed} (active) agents, 
whereas a later analysis by Moses and Tuttle \cite{MT88} (followed by Fagin et al.~\cite{FHMVbook}), for omissions
failure models, and a  more general notion of agreement protocol, 
uses a notion of common belief amongst the \emph{nonfaulty} agents. 
As generally understood, in the crash failure model, an agent may be nonfailed, but still faulty, because it 
will fail at a later time. There exists some gaps in reasoning in these sources related to these issues. 
We clarify the relationship between these notions, both at the level of specifications and the 
knowledge based program. Specifically, we show that both the SBA specification 
and the common belief condition used in the knowledge based program for SBA may refer to either the
nonfaulty or the nonfailed agents, without change of meaning.

There are also some divergences between the formal modelling of the crash failure model 
in the literature. 
Dwork and Moses \cite{DM90} use a distinguished ``crashed'' state to represent when an agent has crashed, 
whereas a later presentation of their results in Fagin et al.~\cite{FHMVbook} 
models crashed agents as simply failing to send messages from 
some point on (making this model a special case of the sending omissions model). 
This turns out to have an impact on the notion of common knowledge that can be used
in these models. In the interests of generality, we develop a general 
modelling of failures that encompasses both of these models of crash failures. 
We are then able to establish an  equivalence between the different notions of 
common knowledge that have been used in the  crash failures case. 

Our general failure model can also represent  sending and receiving omissions failures, 
but is more general than others that have been used in the literature on 
the application of epistemic logic to distributed computing, in that it can also 
represent failures in which messages and agent local states can be corrupted. 

Using the resulting unified understanding of the literature, we then turn to the main
contributions of the paper, in which we consider the knowledge based program $\kbp$ that, when implemented 
with respect to the full information exchange, yields an SBA protocol 
that is an \emph{optimum} with respect to all possible SBA protocols (for a fixed failure model), 
in the sense that no other SBA protocol can decide earlier in any run. 
We study the senses in which we obtain optimality of implementations of $\kbp$ 
with respect to a limited information exchanges. 

We show that, if we implement $\kbp$ with respect to a given information exchange protocol, 
we may also obtain an implementation that is an optimum relative to protocols using that information exchange. 
This result requires several assumptions on the information exchange and failure model. 
With respect to the information exchange, it is assumed that the agents record 
information about actions they have performed separately from their record of communications received, 
and do not  exchange information about the specific actions that they have performed. In particular, 
agents should not inform others about the fact that they have made a decision, or what that decision is. 
With respect to the failure model, it is assumed that state corruption failures act independently on the 
action and communications records, and do not corrupt the action record. (The usual omissions failures models 
satisfy this constraint, since they do not allow state corruption errors.) 

We also show that, with respect to a weaker set of assumptions on the information exchange and failure model, 
implementations of the knowledge-based program $\kbp$ satisfy a weaker sense of optimality. 
We show that 
if we allow agents to also exchange the information that they have taken a decision, but not 
\emph{what} that decision is, then the knowledge based program still yields  an implementation that is \emph{optimal} 
amongst protocols using the given information exchange, in the sense that this implementation cannot 
be improved upon by any SBA protocol using that information exchange. (An SBA protocol $P$ improves upon 
an SBA protocol $P'$ if in every run, $P$ decides no later than $P'$, and there are some runs where $P$ decides earlier. 
In this sense of optimality, there may be several incomparable optimal protocols, whereas an optimum protocol is necessarily unique, modulo the
actual decisions made.) 

Finally, we show by example that, under these weaker assumptions, we do not always get an optimum. 
The example uses the sending omissions model, and illustrates a subtle tradeoff arising in limited information 
exchange settings: the behaviour of faulty agents, once they discover that they are faulty, 
can impact the decision time of the nonfaulty agents. 
According to the knowledge-based program, faulty agents should decide as soon as they discover 
they are faulty. However, depending on the information exchange, it may be advantageous for
such agents to instead continue to behave as if they were nonfaulty, if that increases information 
flow to the nonfaulty agents.

Our motivation for developing these results was work reported elsewhere \cite{AlpturerHM25}, 
in which we have been using automated synthesis techniques to derive a concrete protocol
from a knowledge based program and a description of the limited information exchange and  
failure model in which it operates. 
The results of the present paper help us to understand the precise 
optimality guarantees satisfied by the implementations obtained using this process.

The structure of the paper is as follows. We begin in Section~\ref{sec:intsys} by recalling the general \emph{interpreted systems} semantics for the logic of
knowledge, and introducing the modal operators needed for the work. Section~\ref{sec:sba} states the specification for the Simultaneous Byzantine 
Agreement problem. Section~\ref{sec:ixf} describes how an interpreted system is generated from an underlying information exchange protocol, a model of the 
failures against which the solution needs to defend, and a protocol used by agents to make their decisions. In Section~\ref{sec:crash}, 
we reconsider the knowledge based characterization of SBA in the crash failures model 
due to Dwork and Moses \cite{DM90}, and show how this is related to the later 
characterization of  Moses and Tuttle \cite{MT88} for omissions failures. The upshot of this analysis is that the Moses and Tuttle characterization can 
be applied in all cases. We then apply this characterization to study optimality of SBA protocols with respect to limited information exchanges 
in Section~\ref{sec:optimality}. Section~\ref{sec:example} presents a counter-example showing that the knowledge based characterization 
does not always yield an optimum solution in limited information exchange contexts. Section~\ref{sec:concl} concludes with a discussion of 
related work and open problems.

\section{Knowledge in Interpreted Systems} \label{sec:intsys}

We  use the general semantic model of Fagin et al.~\cite{FHMVbook} 
to model  the semantics of the logic of knowledge. 
We model the \emph{global states} of a distributed system involving $n$ agents from the set $\Agents = \{1, \ldots,n\}$
as a set $L_e \times L_1 \times \ldots \times L_n$, where $L_e$ is a set of states of the environment in 
which the agents operate, and each $L_i$, for $i  \in \Agents$, is a set of \emph{local states of agent $i$}. 
A \emph{run} of the system is a function $r: \Nat \rightarrow L_e \times L_1 \times \ldots \times L_n$
mapping times, represented as natural numbers, to global states.  
A \emph{point} is a pair $(r,m)$ consisting of a run $r$ and a time $m$. 
An \emph{interpreted system} is a pair $\I = (\R,\pi)$ consisting of a set $\R$ of runs and an 
\emph{interpretation} $\pi: \R\times \Nat \rightarrow \mathcal{P}(\Prop)$ 
associating a subset of the set $\Prop$ of propositions to each point of the system.  

The semantics of knowledge is defined using a relation $\sim_i$ on points for each agent $i$, 
given by $(r,m) \sim_i(r',m')$ if $r_i(m) = r_i(m')$.  
The interpreted systems we consider in this paper will generally be \emph{synchronous} 
in the sense that if $(r,m) \sim_i(r',m')$ then $m = m'$. 
For each agent $i$, the logic of knowledge has a modal operator $K_i$, such that 
$K_i\phi$ is a formula for each formula $\phi$. Satisfaction of formulas $\phi$ at points $(r,m)$ of an 
interpreted system $\I= (\R,\pi)$
is defined by the relation $\models$, such that
\begin{enumerate} 
\item $\I,(r,m) \models p$ if $p \in \pi(r,m)$, for atomic propositions $p \in \Prop$, and 
\item $\I,(r,m) \models K_i \phi$ if  $\I,(r',m') \models \phi$  for all points $(r',m') \sim_i (r,m)$. 
\end{enumerate} 
A formula $\phi$ is \emph{valid} in an interpreted system $\I$ if 
$\I,(r,m) \models \phi$ for all points $(r,m)$. 

We work with a number of different notions of group knowledge, that operate with respect to an \emph{indexical set} $S$ of agents, which differs from 
point to point in the system. That is, we assume that there is a function $S$ mapping each point of the system to a set of agents. 
The semantics of the atomic formula $i\in S$ is given by $\I,(r,m) \models i \in S$ if   $i\in S(r,m)$.

An agent may not know whether it is in a set $S$. We can define a notion of belief, relative to the indexical set $S$, by $\bel{S}_i \phi = K_i(i\in S \rimp \phi)$.  
We define the notions of ``everyone in $S$ believes'' and ``everyone in $S$ knows'', by 
$\eb{S} \phi = \bigwedge_{i\in S} \bel{S}_i \phi$ and $\ek{S} \phi = \bigwedge_{i\in S} K_i \phi$. 
Common belief, relative to an indexical set $S$, is defined by 
$\cb{S} \phi = \eb{S} \phi \land \eb{S}^2 \phi \land \ldots$.%
\footnote{
Moses and Tuttle \cite{MT88} define this as $\phi \land \cb{S}\phi$. 
If we write this as $TCB_S( \phi)$ (for ``true common belief)  
we have $TCB_S(\phi)\rimp \phi$ valid even when $S \neq \emptyset$ is not valid. 
However, their application of this operator is for the set $S$ of nonfaulty agents, 
which is always non-empty because they work with the assumption that 
the number $t$ of faulty agents is at most the number of agents minus two.
In all their applications, therefore, $TCB_S( \phi)$ is equivalent to $\cb{S}(\phi)$. 
}
Common knowledge, relative to an indexical set $S$, is  is defined by 
$\ck{S} \phi = \ek{S} \phi \land \ek{S}^2 \phi \land \ldots$.

A more semantic characterization of common knowledge is as follows. 
Define the relations $\sim^*_S$ and $\approx^*_S$ on points of a system $\I$
to the reflexive, transitive closures of the relations 
$\sim_S$ and $\approx_S$ on points given by 
\begin{enumerate} 
\item $(r,m) \sim_S (r',m')$ if there exists $i \in S(r,m)$ such that $(r,m) \sim_i(r',m')$
\item $(r,m) \approx_S (r',m')$ if there exists $i \in S(r,m)\cap S(r',m')$ such that $(r,m) \sim_i(r',m')$
\end{enumerate} 
Then we have that $\I,(r,m) \models \ck{S} \phi$ iff $\I,(r',m') \models  \phi$ for all points $(r',m')$ of $\I$
such that $(r,m) \sim^*_S (r',m')$. 
Similarly, $\I,(r,m) \models \cb{S} \phi$ iff $\I,(r',m') \models  \phi$ for all points $(r',m')$ of $\I$
such that $(r,m) \approx^*_S (r',m')$.

These notions are (greatest) fixed points, satisfying
$\cb{S} \phi \equiv \eb{S} \cb{S}\phi$ and $\ck{S} \phi \equiv \ek{S} \ck{S} \phi$. Provided it is valid that 
$S \neq \emptyset$, we have that 
$\eb{S} \phi \rimp \phi$ and  $\ek{S}\phi \rimp \phi$ and 
$\cb{S}\phi \rimp \phi$ and  $\ck{S}\phi \rimp \phi$ are all valid. These are therefore knowledge-like notions.  
Further, for each of the operators $O\in \{K_i, \bel{S}_i,\eb{S},\ek{S},\cb{S},\ck{S}\}$
we have $O\phi \rimp O\psi$ valid if $\phi\rimp \psi$ is valid.

\begin{proposition} 
\label{prop:ck:contain} 
If $S$ and $T$ are indexical sets such that $S\subseteq T$ is valid, then the formulas
$\bel{T}_i \phi \rimp \bel{S}_i \phi$, 
$\ck{T} \phi \rimp \ck{S} \phi$ and $\cb{T} \phi \rimp \cb{S} \phi$ are valid.
\end{proposition}
\begin{proof} 
Suppose that $\I,(r,m) \models \bel{T}_i\phi$. 
Then $\I,(r',m') \models  \phi$ for all points $(r',m') \sim_i (r,m)$ 
such that $i \in T(r',m')$. 
We show that $\I,(r,m) \models \bel{S}_i\phi$. 
Let $(r',m) \sim_i (r,m)$ and suppose that $i \in S(r',m)$. 
Since $S \subseteq T$ is valid in $\I$, we also have $i \in T(r',m)$, 
and it follows that $\I,(r',m') \models  \phi$.  

Similarly, suppose that $\I,(r,m) \models \ck{T}\phi$. 
Then
$\I,(r',m') \models  \phi$ for all points $(r',m')$ of $\I$
such that $(r,m) \sim^*_T (r',m')$. 
When $S \subseteq T$ is valid in $\I$, 
we have  for all points $(r',m')$ that $(r,m) \sim^*_S (r',m')$ implies 
$(r,m) \sim^*_T (r',m')$, hence $\I,(r',m') \models  \phi$. 
This shows that $\I,(r,m) \models \ck{T}\phi$. 

The proof of $\cb{T} \phi \rimp \cb{S} \phi$ is similar, using instead the 
characterization in terms of the relations $\approx^*_{S}$ and $\approx^*_{T}$.
\end{proof} 

\begin{proposition} \label{prop:k:b} 
The formulas
$K_i \phi \rimp \bel{S}_i \phi$, 
$\ek{S} \phi \rimp \eb{S} \phi$ and 
$\ck{S} \phi \rimp \cb{S} \phi$ are valid.
\end{proposition}

\begin{proof} 
Validity of $K_i \phi \rimp \bel{S}_i \phi$ is immediate from the fact that 
$\bel{S}_i \phi$ is $K_i(i \in S \rimp \phi)$. 

For $\ek{S} \phi \rimp \eb{S} \phi$, 
note that if $\ek{S} \phi$ then $\bigwedge_{i \in S} K_i \phi$, 
which implies $\bigwedge_{i \in S} \bel{S}_i\phi$ by the previous paragraph, 
and this is $\eb{S} \phi$. 

For $\ck{S} \phi \rimp \cb{S} \phi$,
we show by induction that  $\ek{S}^k \phi \rimp \eb{S}^k \phi$  is valid for all $k>0$. 
The base case of $k=1$ is the result of the previous paragraph. 
Assuming $\ek{S}^k \phi \rimp \eb{S}^k \phi$ is valid, we have that 
if $\ek{S}^{k+1} \phi$ then $\ek{S}(\ek{S}^{k} \phi)$, which implies 
$\eb{S}(\ek{S}^{k} \phi)$ by the result of the first paragraph, and 
then $\eb{S}(\eb{S}^{k} \phi) = \eb{S}^{k+1} \phi$ by the inductive hypothesis.
It follows that  $\ck{S} \phi = \bigwedge_{k>0} \ek{S}^k \phi$ 
implies  $\bigwedge_{k>0} \eb{S}^k \phi = \cb{S} \phi$. 
\end{proof}

\section{Simultaneous Byzantine Agreement} \label{sec:sba} 

The specification of Simultaneous Byzantine Agreement concerns a set of agents, operating subject to faults, 
who are  required to reach a common decision on a set of values from some set $\Values$. 
Each agent $i$ starts with a preferred value $\init_i$ for the decision to be made. 
At each moment of  time, each agent $i$ chooses an action from the set $A_i = \{\noop\} \cup \{\decide_i(v) ~|~v \in \Values\}$. 

We may state the specification SBA($S$) of Simultaneous Byzantine Agreement with respect to an indexical set $S$ as follows: 
\begin{itemize} 
\item[] \textbf{Unique-Decision:} Each agent $i$ performs an action $\decide_i(v)$ (for some $v$) at most once.%
\footnote{In Byzantine contexts, with $S$ equal to the set of nonfaulty agents, 
it would be appropriate to change this to say that each agent $i\in S$ performs an action $\decide_i(v)$ (for some $v$) at most once, 
since the condition as stated cannot be guaranteed. However, in benign failure models this stronger condition 
can be easily satisfied.}
\item[] \textbf{Simultaneous-Agreement(S):} If $i\in S$ and $i$ performs $\decide_i(v)$ then, at the same time, 
all $j\in S$ also perform $\decide_j(v)$.
\item[] \textbf{Validity(S):} If $i\in S$ and $i$ performs $\decide_i(v)$ then there exists an agent $j$ with $\init_j = v$. 
\end{itemize} 

There are variances in the literature as to the set $S$ that should be used in this specification. 
Most work takes $S$ to be the set $\N$ of nonfaulty agents. However, 
Dwork and Moses \cite{DM90} (on the crash failure model) appears to refer to the nonfaulty agents, informally, in their 
introduction, but work with the active (nonfailed) agents $\A$ in their proofs. 
We consider the alternatives below in order to clarify these points.

Various termination requirements are considered in the literature. We omit a termination condition in 
order to study the problem in settings with limited information exchange, 
that do not always allow a decision to  be made.

\section{Information Exchange Protocols and Failure Models} \label{sec:ixf}

To model protocols for SBA under a variety of failure models, and study the effect of a range of assumptions about how 
agents in these protocols exchange information, we compose protocols into two parts, 
a \emph{decision protocol} $P$ and an \emph{information exchange} $\exchange$. 
The environment in which the agents operate will be modelled as \emph{failure model} $\failures$. 

An information exchange $\exchange$ associates to each agent $i$ 
a tuple $\exchange_i = \langle L_i, I_i, \Msg_i, \mu_i, \delta_i\rangle$, where 
\begin{enumerate} 
\item $L_i$ is a set of local states for agent $i$; 
\item $I_i\subseteq L_i$ is a set of \emph{initial states}; 
\item $\Msg_i$ is a set of \emph{messages} that agent $i$ may send, assumed to contain the value $\bot$
representing that the agent sends no message; 
\item $\mu_i: L_i \times A_i \rightarrow (\Agents \rightarrow \Msg_i)$ is a function, 
such that $\mu_i(s,a)(j)$ represents the message that agent 
$i$, with local state $s$, sends agent $j$ in a round in which it performs action $a$; 
\item $\delta_i: L_i \times A_i \times \Pi_{j \in Agents} \Msg_j \rightarrow L_i$, is a function, such that 
$\delta_i(s,a,(m_1, \ldots ,m_n))$ represents the local state of  agent $i$ immediately after a round in which 
the agent started in local state $s$, performed action $a$, and received messages $(m_1, \ldots,m_n)$
from agents $1,\ldots, n$ respectively. 
\end{enumerate} 
A decision protocol $P$ for an information exchange $\exchange$ consists of a function $P_i: L_i \rightarrow A_i$ for each agent $i$.

We focus here on \emph{synchronous} protocols in which local states in $L_i$ are of the form 
$\langle \init_i, \Time_i, \ldots\rangle $, where $\init_i\in \Values$ represents agent $i$'s initial preference for the decision 
\begin{tarkabsout} to be made, 
\end{tarkabsout} 
and $\Time_i$ represents the current time. (In the case of the crash failures model, there is also an additional state $\crashed_i$.)
The 
\begin{tarkabsout} update 
\end{tarkabsout} 
function $\delta_i$  acts so that if $\delta_i(\langle \init_i, \Time_i, \ldots\rangle , a,m) = \langle \init'_i, \Time'_i, \ldots\rangle $
then $\init'_i = \init_i$ and $\Time'_i = \Time_i +1$. 

In the \emph{full information} information exchange $\exchange_{FIP}$ for SBA, 
agents' initial local states consist of their initial preferences,  
agents send their complete local states to all other agents in each round,
and update their states by recording all messages received in their local state.
That is, initial states are values $\ \init_i$, for all agents $i,j$, states $s\in L_i$, and actions $a$, we have  $\mu_i(s,a)(j) = s$, 
and $\delta_i(s,a,m) = s\cdot m$ for all message vectors $m$. (The action $a$ and the time  are not recorded explicitly  in the local state, in this model, 
  but can be deduced.)

A failure model is given by a tuple $\failures = \langle L^*_e,I_e, \delta_e, \adv\rangle$ 
\begin{enumerate} 
\item $L^*_e$ is a set of states of the environment.
\item $I_e\subseteq L^*_e$ is a nonempty set of initial states of the environment.
\item $\delta_e : L^*_e \times \Pi_{i \in \Agents} A_i \rightarrow L^*_e$, such that 
$\delta_e(s,(a_1, \ldots, a_n))$ represents how the state of the environment is updated in a round in which 
agents perform actions $a_1,\ldots,a_n$.  (Dependence on agent actions allows the environment to record information about the actions performed by the agents. We could also 
include here a dependence on the messages sent, 
but we will not need this for the failure models considered in this paper.) 

\item $\adv$ is a nonempty set of adversaries, where each adversary is given by a tuple $\langle \Delta^t,  \Delta^r,\Delta^s\rangle$, where 
\begin{itemize} 
\item $\Delta^t: \Nat \times \Agents\times  \Agents \times \bigcup_{i \in \Agents} \Msg_i \rightarrow \bigcup_{i \in \Agents} \Msg_i $ is a function, such that 
$\Delta^t(k,i,j,m)$ is a message resulting from a fault, if any, through which the 
environment perturbs the message $m$ transmitted by agent $i$ to agent $j$ 
in round $k+1$. 
\item $\Delta^r: \Nat \times \Agents\times  \Agents \times \bigcup_{i \in \Agents} \Msg_i \rightarrow \bigcup_{i \in \Agents} \Msg_i $ is a function, such that 
$\Delta^r(k,i,j,m)$ is a message resulting from a fault, if any, through which the 
environment perturbs the message $m$ received by agent $j$ from agent $i$ 
in round $k+1$. 
\item $\Delta^s = (\Delta^s_i: \Nat \times  L_i \rightarrow L_i)_{i \in \Agents}$ is a collection of functions, 
representing effects that faults have on the agents' local states, such  that 
$\Delta^s_i(k, s^*_i) = s'_i$ when the effect of the fault, if any, is to cause 
state $s^*_i$ of agent $i$ to be modified in round $k+1$ to state $s'_i$, for each agent $i$.
(Here we write $s^*_i$ to indicate the state of the agent \emph{after} it has applied its state update function $\delta_i$ for the round.)
\end{itemize} 
\end{enumerate} 

Given a decision protocol $P$, information exchange $\exchange$ and failures model $\failures$, 
we define the interpreted system $\I_{P,\exchange, \failures}= (\R_{P,\exchange, \failures},\pi)$
with global states $L_e \times L_1 \times \ldots \times L_n$, 
where $L_e = L^*_e \times \adv$, 
and the runs $r\in \R_{P,\exchange, \failures}$ are 
defined by 
\begin{enumerate} 
\item $r(0) = ((s_e, (\Delta^t,\Delta^r,\Delta^s)),s_1, \ldots , s_n)$, where 
$s_e \in I_e$ and $s_i \in I_i$ for each $i \in \Agents$, and $(\Delta^t,\Delta^r,\Delta^s) \in \adv$. 
\item for all times $k$, if $r(k) = ((s_e, (\Delta^t,\Delta^r,\Delta^s)),s_1, \ldots , s_n)$,
 then $r(k+1)$ is the state $((s'_e, (\Delta^t,\Delta^r,\Delta^s)),s'_1, \ldots , s'_n)$ obtained as follows. 
 
 For each agent $i$, let $a_i = P_i(s_i)$ be the action selected by the decision protocol, 
 and let $m_{i,j} = \mu_i(s_i,a_i)(j)$ be the message that agent $i$ sends to agent $j$, 
 according to the information exchange $\exchange_i$. 
 
 Note that the adversary $(\Delta^t,\Delta^r,\Delta^s)$ is the same in $r(k)$ and $r(k+1)$. 
 The remaining state of the environment is updated from $s_e$ to 
 $s'_e = \delta_e(s_e,(a_1, \ldots,a_n))$.

 For each agent $i$ and $j$, let $m'_{i,j}= \Delta^r(k,i,j,\Delta^t(k,i,j,m_{i,j}))$ be the message resulting from  
 any faults caused by the adversary in the transmission from $i$ to $j$ 
 (first by application of the function $\Delta^t$ at the sender $i$, and then by application of the function 
 $\Delta^r$ at the receiver $j$.  
 Thus, for each agent $j$, the messages received by agent $j$ are $(m'_{1,j}, \ldots, m'_{n,j})$. 
 The expected effect of these message receptions on the agents' local states 
 is to transition from $(s_1, \ldots , s_n)$ to 
 $(s^*_1, \ldots , s^*_n)$, where $s_j^* = \delta_j(s_j,a_j,(m'_{1,j}, \ldots,m'_{n,j}))$. 
 We define $s'_i= \Delta^s_i(k,s^*_i)$ for each agent $i$. That is, we apply the perturbation 
 $\Delta^s$ to the local states of the agents after they have updated their local states according
 to the information exchange. 

\end{enumerate} 

Agents may experience a number of different types of faults. 
Agent $i$ has a \emph{transmission fault} in round $k+1$ of run $r$ 
if $\Delta^t(k,i,j,m_{i,j}) \neq m_{i,j}$, where $m_{i,j}$ is the message sent by $i$ to $j$ in round $k+1$. 
Agent $j$ has a \emph{reception fault} in round $k+1$ of run $r$ 
if $\Delta^r(k,i,j,m_{i,j}) \neq m_{i,j}$, where $m_{i,j}$ is the message delivered from $i$ to $j$ in round $k+1$. 
Agent $i$ has a \emph{state} fault if, in round $k+1$, we have 
$s_i'= \Delta^s_i(k,s_i^*)$ and $s'_i\neq s^*_i$. 
If none of these types of faults apply, then we say that agent $i$ does not have a fault in 
round $k+1$. 
We say that an agent $i$ is \emph{faulty} in a run $r$ if it has a fault of any type for some round $k\in \Nat$. 
Agent $i$ is \emph{nonfaulty to time $k$} if it does not have a fault in rounds $1 \ldots k$ in run $r$. 
We define the indexical set $\N(r,k)$ to be the set of agents that are not faulty in $r$, 
and the indexical set $\A(r,k)$ to be the set of agents that are not faulty to time $k$. 

We remark that because time is encoded in agents' local states, in order to satisfy the synchrony constraint 
in interpreted systems, we henceforth assume that the state perturbation function $\Delta^s_i$ 
does not modify the value of the $\Time_i$ component of a local state. That is, we assume that 
even faulty agents have reliable clocks.

The interpretation $\pi$ gives meaning to propositions, dependent on the protocol, information exchange and  failure model. In particular,  for agents $i$, values $v\in \Values$, and indexical sets $S,T$,
\begin{itemize} 
\item $\decides_i(v)$ is in $\pi(r,m)$ if $P_i(r,m) = \decide_i(v)$;
\item $i\in S$ is in $ \pi(r,m)$ if $i \in S(r,m)$; 
\item $S \subseteq T$ is in $\pi(r,m)$ if $S(r,m) \subseteq T(r,m)$; 
\item $S = \emptyset$ is in $\pi(r,m)$ if $S(r,m) = \emptyset$; 
\item $\exists v$ is in $\pi(r,m)$ if there exists an agent $i$ with $\init_i = v$ in $r_i(0)$. 
\end{itemize} 

Plainly, $\N \subseteq \A$ is valid; any agent that never fails will not have failed 
before the current time. 
Note that $\N$ is independent of the time, and depends only on the run: 
$\N(r,m) = \N(r,m')$ for all times $m,m'$. This does not hold for 
$\A$. 

A \emph{context for SBA} is a pair $\gamma = (\exchange, \failures)$, where $\exchange$ is an information exchange and $\failures$ is a failure model. 
For brevity we may also write $\I_{P,\gamma}$ for the interpreted system $\I_{P,\exchange,\failures}$.

Commonly studied failure models from the literature can be represented in the above form. 
We say that $\Delta^s$ is \emph{correct} for agent $i$ if $\Delta^s_i(k, s_i) = s_i$ for all $s_i\in L_i$ and $k \in \Nat$. 
Similarly $\Delta^t$ is correct for agent $i$ if $\Delta^t(k,i,j,m) = m$ for all $k,j$ and $m$, and 
$\Delta^r$ is correct for agent $j$ if $\Delta^r(k,i,j,m) = m$ for all $k,i$ and $m$.  
\begin{itemize} 
\item In the \emph{hard crash} failures model of \cite{DM90}, agents may crash at any time. In the round 
in which an agent crashes, it sends an arbitrary subset of the set of messages it was required to send in the round. 
To represent this model, we require that agents'  local state sets $L_i$ contain a distinguished state $\crashed$. 
We always take $\Delta^r$ to be correct for all agents $i$. 
An adversary for which agent $i$ crashes in round $k+1$ has $\Delta^s_i(k,s_i) = \crashed$ for all $s_i\in L_i$, and 
there exists a set $J\subseteq \Agents$ such that, for all messages $m$,
$\Delta^t(k,i,j,m) = \bot$ for $j\in J$, and $\Delta^t(k,i,j,m) = m$ for $j\in \Agents\setminus J$. 
For $k'> k$, we also have  $\Delta^s_i(k',s_i) = \crashed$, and 
$\Delta^t(k',i,j,m) = \bot$ for all agents $j$.
For agents that do not crash, $\Delta^s,\Delta^t$ and $\Delta^r$ are correct. 
We write $\Crash_t$ for the failure model in which $\adv$ contains the adversaries in which 
$t$ or fewer agents may crash. 

\item In the \emph{communications crash} version of the crash failures model used in \cite{FHMVbook}, again agents may crash at any time, and in the round 
in which an agent crashes, it sends an arbitrary subset of the set of messages it was required to send in the round. 
However, we do not require for this model that agents'  local state sets $L_i$ contain the distinguished state $\crashed$. 
Instead, failures in this model can be understood as crashes of the agent's transmitter. 
An adversary for which agent $i$ crashes in round $k+1$ has $\Delta^s$ correct for all $i$, 
and there exists a set $J \subseteq \Agents$ such that, for all messages $m$, 
 $\Delta^t(k,i,j,m) = \bot$ for $j\in J$, and $\Delta^t(k,i,j,m) = m$  for $j\in \Agents \setminus J$, and 
for $k'\geq k$, we also have $\Delta^t(k,i,j,m) = \bot$ for all agents $j$.
In all other cases, $\Delta^s$, $\Delta^t$, and $\Delta^r$ are correct. 
We write $\ComCrash_t$ for the failure model in which $\adv$ contains the adversaries in which 
$t$ or fewer agents may crash. 

\item In the \emph{Sending Omissions} model $SO_t$, $\Delta^s$ and $\Delta^r$ are correct for all agents, but $\Delta^t$ 
may allow failures for up to $t$ agents. 

\item In the \emph{Receiving Omissions} model $RO_t$, $\Delta^s$ and $\Delta^t$ are correct for all agents, but $\Delta^r$ 
may allow failures for up to $t$ agents.

\item In the \emph{General Omissions} model $GO_t$, $\Delta^s$ is correct for all agents, but $\Delta^r$ and $\Delta^t$ 
may allow failures may allow failures for up to $t$ agents.
\end{itemize} 
Other types of failure assumptions can also easily be modelled, such as crashing agents sending messages to a 
\emph{prefix} of the list of agents $[1\ldots n]$, atomic transmission  failures in which a failing agent transmits 
to no other agents, message corruption, etc.

\section{Crash Failures} \label{sec:crash} 

\newcommand{\ccrash}{\gamma_{\mathit{Crash}(t,n)}}

We first consider some subtleties relating to the hard crash failures model and the knowledge based program 
used by Dwork and Moses \cite{DM90}.  This context has 
consequences for the agent's knowledge, and affects the knowledge based program 
developed in \cite{DM90}. In this context, 
$\N$ represents the nonfaulty agents and  the set $\A$ of agents that have not failed to the current time
is the set of \emph{active} agents, that have not yet crashed. 

The specification for 
SBA  for the crash failures model appears to be given by Dwork and Moses 
as SBA($\N$), i.e., with respect to \emph{nonfaulty} agents. On the other hand, it is 
stated by Fagin et al.~\cite{FHMVbook} as SBA($\A$), i.e., for the \emph{nonfailed} agents. 
Moses and Tuttle \cite{MT88} consider omissions failures, and state a specification that is a generalization 
of SBA($\N$)
(to a richer set of coordinated action problems, and allowing the inclusion of a termination requirement). 
The use of $\N$ appears to be the more common approach in the broader literature on distributed algorithms. 
We may note the following relationship between these specifications, 

\begin{proposition} 
Let $\gamma$  be any context for SBA, and $P$ any protocol for this context, 
and let $S$ and $T$ be indexical sets of agents such that $\I_{P,\gamma} \models S \subseteq T$.
If $\I_{P,\gamma} \models \text{SBA}(T)$ then $\I_{P,\gamma} \models \text{SBA}(S)$.
In particular if $\I_{P,\gamma} \models \text{SBA}(\A)$ then $\I_{P,\gamma} \models \text{SBA}(\N)$.
\end{proposition} 
\begin{proof} 
The Unique-Decision is property is independent of the indexical set in the specification, so holds trivially. 
Validity($T$) implies Validity($S$) since $S\subseteq T$ is valid. 
Also, Simultaneous-Agreement($T$) implies Simultaneous-Agreement($S$) for the same reason. 
Thus validity of SBA($T$) implies validity of SBA($S$). 
The fact that $\I_{P,\gamma} \models \text{SBA}(\A)$ implies $\I_{P,\gamma} \models \text{SBA}(\N)$
follows directly from the fact that $\N\subseteq \A$ is valid. 
\end{proof} 

Under certain conditions, we also have a converse to this result. 

\begin{proposition} \label{prop:subsetSBA}
Suppose that $P$ is a protocol for the context $\gamma$, and let $S$ and $T$ be indexical sets of agents in $\I_{P,\gamma}$, 
such that  
\begin{itemize} 
\item[(a)] for all points $(r,m)$, if $\I_{P,\gamma},(r,m) \models i\in T \land j\in T$, then there exists a run $r'$ such that $(r,m) \sim_i (r',m)$ and $(r,m) \sim_j(r',m)$, and 
$\I_{P,\gamma},(r,m) \models i\in S \land j\in S$, and 
\item[(b)] $\I_{P,\gamma} \models  S \subseteq T $, and 
\item[(c)] $\I_{P,\gamma} \models T \neq \emptyset \rimp S \neq \emptyset$.
\end{itemize} 
Then  $\I_{P,\gamma} \models \text{SBA}(S)$ implies  $\I_{P,\gamma} \models \text{SBA}(T)$. 
\end{proposition}

\begin{proof} 
Assume that $\I_{P,\gamma} \models \text{SBA}(S)$. We first show that 
Simultaneous-Agreement($T$) is valid in $\I_{P,\gamma}$.
Suppose $\I_{P,\gamma},(r,m) \models i \in T\land \decides_i(v)$ and let $j \in T(r,m)$. 
We show that $\I_{P,\gamma},(r,m) \models \decides_j(v)$.
By (a), there exists a point $(r',m)$ such that $(r,m) \sim_i(r',m)$ and 
$(r,m) \sim_j(r',m)$ and $\I_{P,\gamma},(r,m) \models i \in S\land j \in S$. 
Since $(r,m) \sim_i(r',m)$, we have $P_i(r'_i(m)) = P(r_i(m)) = \decide_i(v)$, 
so also  $\I_{P,\gamma},(r',m) \models \decides_i(v)$.
Since $\I_{P,\gamma} \models \text{SBA}(S)$ we have Simultaneous-Agreement($S$) 
and it follows  
that $\I_{P,\gamma},(r',m) \models \decides_j(v)$.
Because  $(r,m) \sim_j (r',m)$, we also have $\I_{P,\gamma},(r,m) \models \decides_j(v)$, as required. 

Next, we show Validity($T$) is valid in $\I_{P,\gamma}$. Let $(r,m)$ be a point where 
$\I_{P,\gamma},(r,m) \models \decides_i(v) \land i \in T$. Since Simultaneous-Agreement($T$) is valid in $\I_{P,\gamma}$, 
as shown above, 
we have $\I_{P,\gamma},(r,m) \models \decides_j(v)$ for all $j\in T(r,m)$. Since $S\subseteq T$ is valid, by (b), we have 
$\I_{P,\gamma},(r,m) \models \decides_j(v)$ for all $j\in S(r,m)$. Because 
$S(r,m) \neq \emptyset$, by (c) and the fact that $i \in T(r,m)$, there exists 
$j \in  S(r,m)$ such that $\I_{P,\gamma},(r,m) \models \decides_j(v)$. It now follows from Validity($S$) that 
 $\I_{P,\gamma},(r,m) \models \init_k = v$ for some agent $k$. 
 
The property Unique-Decision is the same in SBA($S$) and SBA($T$), so this is immediate. 
\end{proof}

\begin{corollary} \label{cor:sbaNAcrash}
For crash failures and omissions failure contexts $\gamma$ and protocols $P$, with $\I_{P,\gamma} \models \N \neq \emptyset$, 
we have  $\I_{P,\gamma} \models \text{SBA}(\N)$ implies $\I_{P,\gamma} \models \text{SBA}(\A)$. 
\end{corollary} 

\begin{proof} 
The result follows using Proposition~\ref{prop:subsetSBA} with $S = \N$ and $T = \A$.
Condition (b) in Proposition~\ref{prop:subsetSBA} follows from the definitions of $\N$ and $\A$.
Condition (c) is direct, by assumption.
We show that condition (a) of Proposition~\ref{prop:subsetSBA} 
holds for these failure models. Suppose $\I_{P,\gamma},(r,m) \models i \in \A \land j \in \A$. 
Let $r'$ be the run that is identical to $r$ to time $m$, but in which the adversary is modified so that agents $i$ and $j$ 
never fail after time $m$. Since these agents did not have a failure in run $r$ before time $m$ either, we have 
$\I_{P,\gamma}, (r',m) \models i \in \N \land j \in \N$ as required. Because runs are determined by their initial states, 
the protocol $P$ and the adversary, there is no difference  between $r$ and $r'$ in the adversary 
before time $m$, we have $(r,m) \sim_i (r',m)$ and $(r,m) \sim_j (r',m)$ in particular. 
\end{proof} 

Thus, we have SBA($\N$) is equivalent to SBA($\A)$ in crash and omission failure models when $\N \neq \emptyset$ is valid.
While SBA($\N$)  requires only the nonfaulty agents to decide simultaneously, in fact, 
the stronger statement that all nonfailed agents act simultaneously is implied by this specification.

For a set $S$, write $\decides_S(v)$ for $\bigwedge_{i\in S} \decides_i(v)$.
Dwork and Moses \cite{DM90} Theorem 8 states that for any SBA protocol $P$ for the crash failures model, 
 $\decides_i(v) \rimp \ck{\A}(\decides_\A(v))$ and  $\decides_i(v) \rimp \ck{\A}(\exists v)$ are valid in $\I_{P,\exchange,\Crash_t}$. 
 The specification of SBA is stated informally in the introduction of the paper using the term ``nonfaulty'' 
 but it is not made precise in the paper whether this should be interpreted as referring to the set $\N$ of agents that never fail, 
 or the set $\A$ of active agents, that have not yet failed. The proof of Theorem 8 appears to be using $\A$ as 
 the interpretation. However, the result can also be established using the apparently weaker interpretation $\N$, 
 as shown in the following result.  A second subtlety is that the proof depends on the fact that 
 crash failures have been modelled using the hard crash failures model, 
 so that crashed agents are in a special state $\crashed$, 
 with the property that $P_i(\crashed) = \noop \neq \decide_i(v)$ for all values $v$. 
 
 \begin{proposition} \label{prop:dmCK} 
 Suppose that $P$ is a protocol for the hard crash failures context $(\exchange,\Crash_t)$ with $t<n$ 
 such that $\I_{P,\exchange,\Crash_t} \models \text{SBA}(\N)$. Then 
 $\decides_i(v) \rimp \ck{\A}(\decides_\A(v))$ and  $\decides_i(v) \rimp \ck{\A}(\exists v)$ are valid in $\I_{P,\exchange,\Crash_t}$. 
 \end{proposition} 
 
 \begin{proof} 
 From Proposition~\ref{cor:sbaNAcrash}, we obtain from $\I_{P,\exchange,\Crash_t} \models \text{SBA}(\N)\land \N \neq \emptyset$ 
 that $\I_{P,\exchange,\Crash_t} \models \text{SBA}(\A)$. Suppose that $\I_{P,\exchange,\Crash_t} , (r,m) \models \decides_i(v)$.
 Then we cannot have that $r_i(m) = \crashed$, and thus $i \in \A(r,m)$. It follows from $\text{SBA}(\A)$ that 
  $\I_{P,\exchange,\Crash_t} , (r,m) \models \decides_\A(v)$. Let $j \in\A(r,m)$ and $(r,m) \sim_j (r',m)$. 
  Then $r_j(m) = r'_j(m) \neq \crashed$ so also $j \in\A(r',m)$ and  $\I_{P,\exchange,\Crash_t} , (r',m) \models \decides_j(v)$.
  Using $\text{SBA}(\A)$, we obtain $\I_{P,\exchange,\Crash_t} , (r',m) \models \decides_\A(v)$.This shows that for all agents $i$, 
  $\I_{P,\exchange,\Crash_t} \models\decides_i(v) \rimp \ek{\A} \decides_\A(v)$, which implies that 
 $\I_{P,\exchange,\Crash_t} \models\decides_\A(v) \rimp \ek{\A} \decides_\A(v)$. By induction, this gives 
  $\I_{P,\exchange,\Crash_t} \models\decides_\A(v) \rimp \ck{\A} \decides_\A(v)$, and we derive   $\I_{P,\exchange,\Crash_t} \models\decides_i(v) \rimp \ck{\A} \decides_\A(v)$.
  Next, it follows using  $\I_{P,\exchange,\Crash_t} \models \text{SBA}(\N)\land \emptyset \neq \N \subseteq \A$ and Validity($\N$) that 
  $\I_{P,\exchange,\Crash_t} \models\decides_i(v) \rimp \ck{\A} \exists v$.
 \end{proof} 
 
On this basis,  Dwork and Moses \cite{DM90} use the general knowledge-based program $\kbp(\Phi)$ in which agent $i$ operates as follows
\begin{equation} 
\begin{array}{l} 
\mbox{do $\noop$ until}~ \exists v\in V ( \Phi_{i,v});\\ 
\mbox{for $v$ the least value in $V$ for which}~\Phi_{i,v} ~\mbox{do}~\decide_i(v);\\ 
\mbox{do $\noop$ forever}
\end{array} \label{usimkbp} 
\end{equation} 
where $\Phi$ is a collection of (knowledge-based) formulas indexed by an agent $i$ and  a value
 $v \in \Values$.  
 Dwork and Moses work with the instance of $\kbp(\Phi)$ in which 
$\Phi_{i,v}$ is the  condition for each agent $i$ and possible choice $v\in \Values$
given by  $K_i \ck{\A}(\exists v)$, but we define the knowledge-based program more abstractly using $\Phi$, 
in order to consider alternatives.

A concrete protocol $P$ \emph{implements} $\kbp(\Phi)$ 
with respect to a context $\gamma$ if at all points $(r,m)$ of $\I_{P,\gamma}$, and all agents $i$, 
the action $P_i(r_i(m))$ is the same action as would be selected by agent $i$'s program in 
$\kbp(\Phi)$ at $(r,m)$, with $\Phi_{i,v}$ 
interpreted as true iff $\I_{P, \gamma},(r,m) \models \Phi_{i,v}$. 
 
A few points are worth noting concerning this definition. First, the quantification is over \emph{all} 
agents $i$, so both the faulty and non-faulty agents are required to behave as the knowledge-based program describes. This would not be appropriate in Byzantine settings,  where we cannot expect that 
malicious agents will behave as required. Moreover, even in benign settings, in order to satisfy the 
Unique Decision property, once a decision has been made, 
an SBA protocol needs to be able to deduce from the agent's local state that a decision has 
already been made, so that the appropriate action is $\noop$. 

In general, failure contexts 
in which the state perturbation function $\Delta^s$ may perturb agents' local states may 
interfere with faulty agents' ability to satisfy the Unique Decision property. 
Satisfying Unique Decisions generally requires some assumptions on the action of $\Delta^s$, 
such as the existence of a bit in agents' local states that records whether a decision has been made, 
that is not affected by $\Delta^s$. 

Because of these concerns, in some strong failure settings, 
it would be appropriate for the Unique Decision property to 
weakened to apply to non-faulty agents only, and for the definition of  knowledge-based program 
implementation to be similarly weakened. We leave an exploration of such weakening to be explored elsewhere, and focus in this paper on failure settings where the Unique Decision property 
can be satisfied for all agents.

As already noted, Dwork and Moses \cite{DM90} work with an instance of $\kbp(\Phi)$ in which 
$\Phi_{i,v}$ is the  condition for each possible choice $v\in \Values$
given by  $K_i \ck{\A}(\exists v)$. 
By contrast, Fagin et al \cite{FHMVbook} show that for an SBA($\A$) protocol, 
the formula $\decide_i(v) \rimp \bel{\A}_i \cb{\A} \exists v$ is valid in $\I_{P,\exchange,\failures}$.
On the basis of this, they use $\Phi_{i,v} = \bel{\A}_i \cb{\A} \exists v$ in the  knowledge based program $P(\Phi)$.%
\footnote{The use of a common belief rather than a common knowledge formula originates in the 
work of Moses and Tuttle \cite{MT88}, but they show just that $i \in \N \land \decide_i(v) \rimp \cbn \exists v$ is valid, and 
write a program in which the condition ``test for  $\cb{\N} \exists v$'' is used. This
work predated the formal definition of knowledge based programs, which requires that the conditions of the program 
be local to an agent. The treatment of \cite{FHMVbook} is therefore more satisfactory.}
In fact, this result holds more generally: 

\begin{lemma} \label{lem:belcbn}
Let $S$ be an indexical set  of agents and suppose that
 $P$ an SBA($S$) protocol for 
an information exchange protocol $\exchange$ and failure environment $\failures$. 
Then the formula $ \decide_i(v) \rimp \bel{S}_i \cb{S} \exists v$ is valid in $\I_{P,\exchange,\failures}$.
\end{lemma} 

\begin{proof} 
For brevity, we write $\I$ for $\I_{P,\exchange,\failures}$. 
We first show that $\decides_S(v) \rimp \cb{S} \decides_S(v)$ is valid in $\I$.
Suppose that  $\I,(r,m) \models \decides_i(v)$. 
Suppose $(r,m) \sim_i (r',m)$ where $i \in S(r',m)$. Then
 $P_i(r'_i(m)) =  P_i(r_i(m)) = \decide_i(v)$. Since $P$ is an SBA($S$) protocol and $i \in S(r',m)$, it follows 
by Simultaneous-Agreement($S$) that 
$\I,(r',m) \models \decides_S(v)$.  This shows that $\I \models \decide_i(v) \rimp \bel{S}_i \decide_S(v)$. 
Since this holds for all $i$, it follows that $\decides_S(v) \rimp \eb{S} \decides_S(v)$ is valid in $\I$. It follows by induction that 
$\decides_S(v) \rimp \cb{S} \decides_{S}(v)$ is valid in $\I$. 
By validity of $ \decide_i(v) \rimp \bel{S}_i \decide_S(v)$, 
we also have that 
$\decides_i(v) \rimp  \bel{S}_i \cb{S} \decides_{S}(v)$ is valid.

Next, notice that $(S\neq \emptyset \land \cb{S}\decide_S(v)) \rimp \eb{S} (S\neq \emptyset \land \cb{S}\decide_S(v) \land \exists v)$ is valid in $\I$. 
This is because (i) $\bel{S}_i(S\neq \emptyset)$ is valid by definition of  $\bel{S}_i$, 
(ii) $\cb{S}\phi \rimp \eb{S} \cb{S} \phi$ is valid for all $\phi$, and because (iii) $(S\neq \emptyset \land \decide_S(v)) \rimp \exists v$ is valid in $\I$ 
by Validity(S). By induction, we conclude that $(S\neq \emptyset \land \cb{S}\decide_S(v)) \rimp \cb{S} \exists v$ is valid in $\I$. 

Combining this conclusion with the validity of 
$\decides_i(v) \rimp  \bel{S}_i \cb{S} \decides_{S}(v)$ from the first paragraph, 
and the validity of $\bel{S}_i(S\neq \emptyset)$, 
we obtain the validity of $ \decide_i(v) \rimp \bel{S}_i \cb{S} \exists v$. 
\end{proof} 

Using this result, we may show the correctness of the knowledge based program for 
the specification $SBA(S)$, for  general information 
exchanges, failure models, and choices of indexical set $S$, subject to a constraint on $S$. 
We say that $S$ is \emph{monotonically decreasing} if for all runs $r$ and times $m> m'$, 
we have $S(r,m) \subseteq S(r,m')$. This is plainly satisfied by $S= \A$ and $S = \N$. 

\begin{proposition} \label{prop:implementationIsSBA}
Let the indexical set $S$ be monotonically decreasing. 
Let protocol $P$ be an implementation of the knowledge based program $\kbp(\Phi)$ 
with respect to an information exchange $\exchange$ and failure model $\failures$, where 
$\Phi_{i,v} = \bel{S}_i \cb{S} \exists v$. Then $\I_{P,\exchange,\failures} \models SBA(S)$.
\end{proposition} 

\begin{proof} 
Let $\I = \I_{P,\exchange,\failures}$. 
Satisfaction of Unique-Decision is immediate from the structure of $\kbp(\Phi)$, since agent $i$ always 
performs $\noop$ after the first time, if any, it performs an action $\decide_i(v)$. 

For Simultaneous-Agreement($S$), consider a run $r$ of $\I$
and let $m$ be the earliest time at which any agent $i\in S(r,m)$ performs $\decide_i(v)$, for some $v$. 
Since $P$ implements  $\kbp(\Phi)$, we have that $v$ is the least value for 
which $\I,(r,m) \models \bel{S}_i \cb{S} \exists v$. 
Since $i\in S(r,m)$, it follows that $\I,(r,m) \models \bel{S}_i \cb{S} \exists v$, 
and hence that $\I,(r,m) \models \bel{S}_j \cb{S} \exists v$ for all $j\in S(r,m)$. 
Since $P$ implements  $\kbp(\Phi)$, this implies that all $j\in S(r,m)$, that have 
not previously decided, also decide on some value at time $m$. 
Indeed, for all agents $j\in S(r,m)$, the value $v$ must be the least value $v'$ 
such that $\I,(r,m) \models \bel{S}_j \cb{S} \exists v'$, for else, by the same argument,  
$\I,(r,m) \models \bel{S}_i \cb{S} \exists v'$ also, contradicting the choice of $v$. 
Thus, all agents $j\in S(r,m)$ that have not previously decided, perform $\decide_j(v)$ at time $m$. 
Finally, note that no agent $j\in S(r,m)$ can have decided at a time $m'<m$, for then
we would have that $i \in S(r,m) \subseteq S(r,m')$ would decide at time $m'$ instead, by 
a similar argument. Thus, in fact, all agents $j\in S(r,m)$ decide at time $m$. 

For Validity, suppose that  $i\in S(r,m)$ performs $\decide_i(v)$ at time $m$ in run $r$ of $\I$. 
Then $\I,(r,m) \models \bel{S}_i \cb{S} \exists v$. Since $i\in S(r,m)$, it follows that 
$\I,(r,m)\models \exists v$. 
\end{proof}

Beyond the use of $\ck{\A}$ instead of $\cb{\A}$ to characterize the conditions for an agent to decide, 
a further difference in the results of \cite{DM90} and \cite{FHMVbook} is the 
modelling of crash failures. Whereas \cite{DM90} uses the hard crash model, \cite{FHMVbook} uses the communication crash model. 
We now clarify the connection between these characterizations: in hard crash contexts, the two characterizations 
are equivalent. 

\begin{proposition} \label{prop:dwmt}
If $P$ is an SBA($\N$) protocol for the hard crash context $(\exchange,\Crash_t)$ with $t<n$ 
then 
$\ck{\A}(\decides_\A(v)) \dimp \cb{\N}(\decides_\N(v))$
and $i\in \A \rimp (K_i \ck{\A}(\decides_\A(v)) \dimp \beln_i \cb{\N}(\decides_\N(v))$
are valid in $\I_{P,\exchange,\Crash_t}$. 
\end{proposition}

\begin{proof} 
Suppose $\I_{P,\exchange,\Crash_t},(r,m) \models \ck{\A}(\decides_\A(v))$. Then we have  $\I_{P,\exchange,\Crash_t},(r,m) \models\cb{\A}(\decides_\A(v))$ by Proposition~\ref{prop:k:b}. Since $\N \subseteq \A$ is valid, it follows that $\I_{P,\exchange,\Crash_t},(r,m) \models\cb{\N}(\decides_\A(v))$, and also that $\decides_\A(v) \rimp \decides_\N(v)$ is valid. It follows 
 that $\I_{P,\exchange,\Crash_t},(r,m) \models\cb{\N}(\decides_\N(v))$. 

Conversely, suppose that  $\I_{P,\exchange,\Crash_t},(r,m) \models\cb{\N}(\decides_\N(v))$. 
Since $\N \neq 0$ is valid, we have $\I_{P,\exchange,\Crash_t},(r,m) \models \decides_i(v)$ for some $i\in \N(r,m)$.
By Proposition~\ref{prop:dmCK}, $\I_{P,\exchange,\Crash_t},(r,m) \models \ck{\A} \decides_\A(v)$. 

This shows validity of $\ck{\A}(\decides_\A(v)) \dimp \cb{\N}(\decides_\N(v))$.
Validity of the formula $i\in \A \rimp (K_i \ck{\A}(\decides_\A(v)) \dimp \beln_i \cb{\N}(\decides_\N(v))$
follows from this using the fact that $K_i \phi \rimp \beln_i \phi$ is valid, and that in the hard crash system 
$\I_{P,\exchange,\Crash_t}$, we have $i \in \A \rimp K_i (i \in \A)$. 
\end{proof} 

Proposition~\ref{prop:dwmt} establishes that, in hard crash contexts, the knowledge based program using 
$K_i \ck{\A}(\decides_\A(v))$ is equivalent to the knowledge based program using $\beln_i \cb{\N}(\decides_\N(v))$, 
since these formulas are equivalent for active agents, and agents that have crashed take no actions in either case. 
However, these knowledge based programs may behave differently in a ``communications crash'' model, 
where crashed agents continue to take actions, since then $i \in A \rimp K_i (i \in A)$ is no longer valid, and 
whether a crashed agent satisfies $K_i (i \not \in \A)$ depends on the information exchange.

Since the characterization of  \cite{FHMVbook} is more general, in the sequel, we work with their belief based decision condition in the knowledge based program $\kbp(\Phi)$,  
assume $\N\neq \emptyset$ is valid, and take SBA($\N$) to be meaning of the specification of SBA. 

However, we may also note that, similar to the equivalence at the level of the specification, the choice of 
$\N$ or $\A$ in the condition of the knowledge based program makes no difference to the semantics. 
Define a \emph{synchronous epistemic bisimulation} on $\I$ with respect to a set of atomic propositions 
$\Prop$ to be a relation $\approx$ such that  
whenever $(r,m) \approx (r',m')$, we have
\begin{itemize} 
\item $m = m'$, 
\item for all $p \in \Prop$, $\I,(r,m) \models p$ iff $\I,(r',m) \models p$, and 
\item for all $i \in \Agents$, $(r,m) \sim_i (r',m)$.  
\end{itemize} 

\begin{proposition} \label{prop:CBtestEquiv}
Suppose that $S$ and $T$ are indexical sets of agents  in an interpreted system $\I$, 
and let $\approx$ be a synchronous epistemic bisimulation on $\I$ with respect to $\Prop$ such that 
(a) $\I \models S\subseteq T$, and 
(b) for all points $(r,m)$ of $\I$ there exists a point $(r',m)$ such that $(r,m) \approx (r',m)$ and 
$S(r',m)= T(r,m)$. 
If $p \in \Prop$ then $\I \models \bel{S}_i \cb{S} p  \dimp \bel{T}_i\cb{T} p$. 
\end{proposition} 

\begin{proof} 
 We have $\I \models \cb{T} \phi \rimp \cb{S} \phi$ and hence 
$\I \models  (\bel{T}_i\cb{T} \phi) \rimp \bel{S}_i \cb{S} \phi$
by Proposition~\ref{prop:ck:contain}. 
For the converse, we prove $\I \models (\neg  \bel{T}_i\cb{T} \phi) \rimp \neg \bel{S}_i \cb{S} \phi$. 
Suppose that $\I,(r,m) \models  \neg \bel{T}_i \cb{T} \phi$. Then there exists a point 
$(r^0,m)\sim_i (r,m)$ such that $i\in T(r^0,m)$ and $\I,(r^0,m) \models \neg \cb{T} \phi_T$. 
Moreover, from the latter we have that there exists a sequence 
$ (r^0,m) \sim_{i_1} (r^1,m) \sim_{i_2} \ldots \sim_{i_k} (r^k,m)$ such that 
$\I,(r^k,m) \models \neg p$ and 
for $j = 1\ldots k$ we have $i_j \in T(r^{j-1},m)\cap T(r^{j},m)$. By the assumptions on $\approx$, 
there exists for each $j = 0 \ldots k$ a run $\rho^j$ of $\I$  such that $(r^j,m) \approx(\rho^j,m)$, 
and $S(\rho^j,m) = T(r^j,m)$. 
Since $(r^k,m) \approx (\rho^k,m)$ we obtain that 
$\I,(\rho^k,m) \models \neg p$. 
Also for $j = 1\ldots k$ we have $i_j \in T(r^{j-1},m)\cap T(r^{j},m) = S(r^{j-1},m)\cap S(r^{j},m) $.
It follows that $\I,(\rho^0,m) \models \neg \cb{S} p$. 

Moreover, 
$i \in T(r^0,m) = S(\rho^0,m)$, 
and because $(r^0,m) \approx(\rho^0,m)$, we have 
$(r^0,m) \sim_i(\rho^0,m)$. 
Because  $(r,m) \sim_i (r^0,m)$, we obtain $(r,m) \sim_i (\rho^0,m)$. 
It follows that $\I,(r,m)\models \neg \bel{S}_i \cb{S} p$. 
\end{proof}

We remark that the above proof does not show that $\I \models \cb{S}\phi \dimp \cb{T}\phi$.

\begin{corollary}
If $p$ is an atomic proposition that depends only on the local states of the agents, and 
$\failures$ is either a crash or omission failure model, then 
$\I_{P,\exchange,\failures} \models (\bel{\N}_i \cb{\N} p) \dimp \bel{\A}_i \cb{\A} p$
\end{corollary} 

\begin{proof} 
Define the relation $\approx$ on the points of $\I_{P,\exchange,\failures}$ by
$(r,m) \approx (r',m)$ if 
for all agents $i$, we have that $i$ has the same initial state in $r$ as in $r'$, and 
the behaviour of the adversary of $r$ up to time $m$ is the same as the 
the behaviour of the adversary of $r'$ up to time $m$.  
In particular, it follows from $(r,m) \approx (r',m)$ that we have $\A(r,m) = \A(r',m)$ and $(r,m) \sim_i (r',m)$ for all agents $i$.
If we take $S = \N$ and $T = \A$ then the assumptions of Proposition~\ref{prop:CBtestEquiv} are satisfied
with respect to $\approx$. 
In particular, note that we can obtain the run $r'$ required for condition (b) by changing the 
adversary so that there are no new faults after time $m$. The claim is then immediate.
\end{proof} 

Taking $p = \exists v$, we see that we can use either the formula $\bel{\N}_i \cb{\N} \exists v$ 
or $\bel{\A}_i \cb{\A} \exists v$ in the knowledge based program, without changing its semantics.

\section{Optimality with Respect to Limited Information Exchange} \label{sec:optimality} 

We now turn to the question of optimality of SBA protocols with respect to limited information exchange. 
In view of the results of the previous section, we may work with indexical set $S$ equal to 
the set $\N$ of non-faulty agents, and take $\Phi_{i,v} = \bel{\N}_i \cb{\N} \exists v$ in the 
knowledge-based program $\kbp(\Phi)$.

The literature has concentrated on implementations $P$ of the knowledge based program $\kbp(\Phi)$ with 
respect the full information exchange, because it can be shown that such implementations $P$ are 
an optimum, in the sense that for every SBA protocol $P'$ using any other information exchange, 
in every run the nonfaulty agents decide using $P$  no later than they would  in the corresponding run of $P'$. 
Here, a run $r$ of $P$ is said to \emph{correspond} to a run $r'$ of $P'$ if they have the same initial global state, 
hence the same adversary and initial states of all the agents. 

Since the full information protocol $P$ may be impractical or even require agents to perform intractable computations,
we are interested in alternative limited information exchanges. However, having selected an information exchange, 
it is still desirable to use a protocol that is optimal amongst those that  use the same information exchange. 
In this section, we consider whether the knowledge based program $\kbp(\Phi)$ 
yields such implementations. We show that this is the case in two distinct senses, subject to some assumptions
about the information exchange. 

In order to fairly compare two decision protocols relative to an information exchange, it helps to assume that the information 
exchange does not explicitly transmit  information about what decisions have been taken. 
Say that an information exchange protocol $\exchange$ with action sets $A_i= \{\noop\} \cup \{\decide_i(v)|v \in \Values\}$
\emph{does not transmit decision information} if for all agents $i$, local states $s\in L_i$, and actions $\decide_i(v_1), \decide_i(v_2) \in A_i$, we have
\begin{itemize} 
\item  $\mu_i(s,\decide_i(v_1)) = \mu_i(s,\decide_i(v_2))$, and 
\item for all message vectors $m$,  we have 
$\delta_i(s,\decide(v_1),m) = \delta_i(s,\decide(v_2),m)$. 
\end{itemize}

We work with the following  order on decision protocols: 
$P' \leq_{\exchange,\failures} P$ if for all runs $r'$ of $\I_{P',\exchange,\failures}$, and all agents $i$, if agent $i$ decides in round $m$ 
in run $r'$, then in the corresponding run $r$ of  $\I_{P,\exchange,\failures}$, agent $i$ decides no earlier than round $m$ (or not at all). 
An SBA protocol $P$ is \emph{optimal} with respect to an information exchange $\exchange$ and failure model $\failures$, if 
for all SBA protocols $P'$ with respect to  $\exchange$ and $\failures$, if $P' \leq_{\exchange,\failures} P$ then $P \leq_{\exchange,\failures} P'$. 
That is, there is no SBA protocol $P'$ that decides no later than $P$, and sometimes decides earlier.

\begin{theorem} \label{thm:sbaopt}
Suppose the information exchange $\exchange$ is synchronous and does not transmit decision information,
and that the protocol $P$ implements $\kbp(\Phi)$ with respect to information exchange $\exchange$ and failure model $\failures$.
Then $P$ is an optimal SBA protocol with respect to information exchange $\exchange$ and failure model $\failures$.
\end{theorem}

\begin{proof} 
We prove optimality. Suppose that $P' \leq_{\exchange,\failures} P$. We show that there is no run where 
some agent $i$ running $P'$ decides strictly earlier than in the corresponding run of $P$. Moreover, 
we show that for all runs $r'$ of $\I_{P',\exchange,\failures}$, and all times $m$, 
then for the corresponding run $r$ of  $\I_{P,\exchange,\failures}$, 
we have that $r'_i(m) = r_i(m)$ for all agents $i$.

The proof is by induction on $m$. For $m=0$, we have that 
$r'_i(0) = r_i(0)$ for all agents $i$, by definition of correspondence.  
Moreover, there can be no instance of $P'$ deciding in an earlier round than $P$ before time $m=0$.

For the inductive case, assume that that we have that for all agents $i$, $r'_i(k) = r_i(k)$ for all $k \leq m$, 
and  there is no instance, before time $m$, of some agent using $P'$ deciding in an earlier round than it would using $P$. 
We show that for each agent $i$, protocols $P'$ and $P$ either both decide (possibly on different values), 
or both perform $\noop$.  It will follow from this that $r'_i(m+1) = r_i(m+1)$ for all agents $i$.
Also, it remains true for each agent $i$ that $P'$ has not decided earlier than $P$ to time $m+1$. 

We first show that for each agent $i$, either both $P_i(r_i(m)) = P'_i(r'_i(m)) = \noop$ or 
there exists $v,v' \in \Values$ such that $P_i(r_i(m)) =\decide_i(v)$ and $ P'_i(r'_i(m))=\decide_i(v')$.  
Obviously, this holds if $P_i(r_i(m)) = P'_i(r'_i(m)) = \noop$, so we need only consider the cases
where either protocol decides.  If $P'(r'_i(m)) = \decide_i(v')$, then by Lemma~\ref{lem:belcbn}, 
we have that $\I_{P',\exchange,\failures}, (r',m) \models \beln_i(\cbn \exists v')$.
Because the local states of corresponding runs of $\I_{P,\exchange,\failures}$
are identical to time $m$ to those of $\I_{P',\exchange,\failures}$, 
it follows that $\I_{P,\exchange,\failures}, (r,m) \models \beln_i(\cbn \exists v')$.
Because $P' \leq_{\exchange,\failures} P$, agent $i$ has not yet decided at the 
point $(r,m)$. Since $P$ implements $\kbp(\Phi)$, it follows that $P_i(r_i(m)) = \decide_i(v)$ for some value $v$. 
Alternately, if $P(r_i(m)) = \decide_i(v)$, then because $P'$ has not decided earlier, and  $P' \leq_{\exchange,\failures} P$, 
we must have $P'(r'_i(m)) = \decide_i(v')$ for some value $v'$. 
Thus, in either case we have that both protocols decide, as required.  

Next, we show that $r'_i(m+1) = r_i(m+1)$ for all agents $i$.
The proof considers several cases, but in each case, 
the fact that the local states of all agents are identical in $r'(m)$ and $r(m)$ 
and that for each agent $i$, protocols $P'$ and $P$ either both decide, or both perform $\noop$, 
implies that the same messages are sent by each agent in round $m+1$ of $r'$ and $r'$. 
(In the case that both protocols decide, we use the fact that $\exchange$ does not transmit decision information.) 
Moreover, the failure patterns are identical in these corresponding runs,  so the 
same vector $\rho_i$ represents the  messages received by agent $i$ in round $m+1$ in run $r$ and in run $r'$. 
If $P_i(r_i(m)) = P'_i(r'_i(m)) = \noop$, 
then we have $r'_i(m+1) = \Delta^s_i(\delta_i(r'_i(m),\noop,\rho_i)) = 
\Delta^s_i(\delta_i(r_i(m),\noop,\rho_i)) = r_i(m+1)$.
Alternately, if $P_i(r_i(m)) =\decide_i(v)$ and $ P'_i(r'_i(m))=\decide_i(v')$ then, 
because $\exchange$ does not record decision information, 
we have $r'_i(m+1) = \Delta^s_i(\delta_i(r'_i(m),\decide(v'),\rho_i)) = 
\Delta^s_i(\delta_i(r_i(m),\decide(v'),\rho_i)) = 
\Delta^s_i(\delta_i(r_i(m),\decide(v),\rho_i)) = r_i(m+1)$. 
\end{proof} 

Note that Theorem~\ref{thm:sbaopt} does not state that an implementation $P$ of the knowledge-based program is an \emph{optimum} SBA protocol, 
in the sense that $P \leq_{\exchange,\failures} P'$ for all SBA protocols $P'$ with respect to $\exchange$  and $\failures$. 
In fact, this is not necessarily true, as we will show by a counter-example, presented in Section~\ref{sec:example}. 
The counter-example illustrates a trade-off between information exchange and decision time: 
sending less information may result in making later decisions. 
The information exchange in this counter-example does not transmit information about particular 
values that have been decided, but it does transmit the information that an agent has decided. 
The example involves a situation in which a faulty agent decides earlier than the nonfaulty 
agents, who may then take this into account. 

To obtain the conclusion that implementations of the knowledge-based program 
yield \emph{optimum} SBA protocols, we need the following stronger assumption about the information exchange. 
Say that an information exchange protocol $\exchange$ \emph{does not transmit information about actions}
if for all agents $i$, there exist sets $S_i, D_i$ such that local states have the form $L_i = S_i \times D_i$, and
\begin{itemize} 
\item  there exists a function $\mu'_i: S_i \rightarrow  \Pi_{j \in Agents} \Msg_j $ 
such that for all $(s,d) \in S_i\times D_i$, and $a\in A_i$, we have $\mu_i((s,d),a) = \mu_i'(s)$, and 

\item 
there exist functions $\delta_i^1 :S_i \times  \Pi_{j \in Agents} \Msg_j \rightarrow S_i$ 
and $\delta_i^2:D_i \times A_i \rightarrow D_i$
such that for all states  $(s,d) \in S_i\times D_i$, and $a\in A_i$
and message vectors $m\in  \Pi_{j \in Agents} \Msg_j $, we have 
$\delta_i ( (s,d), a, m) = (\delta_i^1(s,m), \delta_i^2(d,a))$.

\end{itemize}
Intuitively, this says that in a local state $(s,d)\in L_i$, the local state component 
$s$ depends only on its initial value and the messages that the agent has received, and 
the local state component $d$ depends only on its initial value and the actions that the agent
has performed.  Accordingly, when these conditions hold, we call the local state 
component $s$ the \emph{message memory} 
and the local state component $d$ the \emph{action memory}. 
Similarly, the messages transmitted depend only on the local state component $s$. 
Since component $s$ does not depend on actions that the agent has performed, the messages 
sent by an agent carry no information about the agent's current or past actions. 
(We remark that the variable $\Time_i$, required to ensure synchrony, could be included in 
either the message memory $s$ or the action memory $d$. We later discuss 
an assumption on the failure model, in the context of which it is most natural to choose the latter option.)

One reason for structuring the local states in this form is that it allows us to resolve a tension between 
two concerns: to satisfy the Unique Decision property, we need agents to be able to record the fact that they have decided in their local state, but on the other hand, we want to express that agents transmit no information about their actions. 
If we allowed the messages transmitted to depend on the entire local state, then messages would carry potentially information about actions. 

To ensure that agents can implement the unique decision property, we 
say that the information exchange protocol \emph{records decision information} 
if the set $D_i$ is the disjoint union of sets $D_i^1,D_i^2$, such that (1) 
the initial states $I_i$ of agent $i$ are a subset of $S_i\times D_i^1$, 
(2) for all $d \in D_i^1$, we have $\delta_i^2(d,\noop) \in D_i^1$ and 
$\delta_i^2(d,\decide_i(v)) \in D_i^2$, and (3) for all 
$d \in D^2_i$ and $a \in A_i$, we have $\delta_i^2(d,a) \in D_i^2$.
This states, intuitively, that the agent records whether it has already decided in the state component $d$. 
We can determine whether agent has already 
decided by testing whether $d\in D_i^1$ or $d \in D_i^2$.  By ensuring that 
$P_i((s,d)) = \noop$ if $d \in D_i^2$, this allows the protocol to guarantee that the Unique Decision property will be satisfied. 

Clearly, if $\exchange$ does not transmit information about actions, then it does not transmit decision information. 
Note also that an \emph{early stopping} protocol \cite{DolevRS90}, 
which stops transmitting information once it has 
decided, satisfies the property of not transmitting decision information, but such a protocol may transmit information about actions, 
since we may still have $\mu_i(s,\noop) \neq \mu_i(s,\decide_i(v))$.

We remark that a protocol, as defined by Moses and Tuttle~\cite{MT88}, determines the messages to be sent, and actions to be 
performed, as a function of a \emph{view} (corresponding to our notion of local state) that is comprised 
of a history of messages received, a history of other inputs from the environment, the time, and the agent identity.  
This means that the  protocols of Moses and Tuttle \cite{MT88} (including their full-information protocols) do not transmit information about actions.
However, in the case of a full-information protocol, and other protocols that exchange sufficient information, 
it is in fact possible, knowing the decision protocol that the agents are running, for an agent to deduce what actions other agents have taken in the past. 

To obtain optimum protocols as implementations of the knowledge-based program, we also 
need some further assumptions on the failure model. When the information exchange does 
not transmit information about actions, we say that state failures in the failure model $\failures$ \emph{act independently on message and action memory} if all the state 
perturbation functions $\Delta^s$ in the failure model have the property that 
for all agents $i$, there 
exists functions $\Delta^s_1:\Nat \times S_i \rightarrow S_i$ and $\Delta^s_2:\Nat \times D_i$ such that  $\Delta^s_i(k,(t,d)) = (\Delta^s_1(k, t), \Delta^s_2(k,d))$ for all times $k$ and local states $(t,d)\in L_i=S_i\times D_i$.

\begin{lemma} \label{lem:coreq} 
Suppose that information exchange $\exchange$ does not transmit information about actions, and that $P$ and $P'$ are decision protocols
with respect to $\exchange$. Let the failure model $\failures$ act independently on message and action memory.
Suppose that $r$ and $r'$ are corresponding runs of $\I_{P, \exchange,\failures}$ and $\I_{P',  \exchange,\failures}$ respectively. 
Then 
\begin{enumerate} 
\item For all agents $i$ and all times $m$, the message memories are identical in $r$ and $r'$. That is, 
if $r_i(m) = (s,d)$ and $r'_i(m) = (s',d')$, then $s = s'$. Also, the vector of messages (before and after perturbation by the failure model) sent by agent $i$  in round $m+1$ in $r$ 
is the same as the vector of messages sent by agent $i$ in round $m+1$ of run $r'$. 
\item For all agents $i$ and times $m$ such that agent $i$ has performed only $\noop$ to time $m$ in both $r$ and $r'$, 
we have $r_i(m) = r'_i(m)$. 
\end{enumerate}
\end{lemma} 

\begin{proof} 
 For $k \in \Nat$, let  $r_i(k) = (s_k,d_k)$ and $r'_i(k) = (s'_k,d'_k)$. 
We first show by induction that for all times $k$ and all agents $i$, we have $s_k= s'_k$. 
Since the information exchange does not transmit information about actions, it is immediate from this that 
$\mu_i(r_i(k),P_i(r_i(k))) = \mu'_i(r_i(k)) = \mu'_i(r'_i(k)) = \mu_i(r'_i(k),P'_i(r'_i(k)))$,
so the vector of messages that  agent $i$ sends in round $k+1$ 
(before perturbation by the failure model) is the same in $r$ and $r'$. 
Since the runs are corresponding, the transmission perturbation function $\Delta^t$ is the 
same in $r$ and $r'$, so the messages transmitted are also the same after perturbation. 

For $k= 0$, the claim is immediate from the fact that $r$ and $r'$ correspond, which means that 
$r(0) = r'(0)$. Assuming that the claim holds for $k$, 
as already noted, the vector of messages that agent $i$ sends in round $k+1$, after perturbation by $\Delta^t$, 
is the same  in $r$ and $r'$. 
Since the runs correspond, each agent $i$ also receives, in round $k+1$, 
both before and after perturbation by $\Delta^r$, 
the same vector $v_i$ of messages in both $r$ and $r'$. 
Since state failures act independently on the message and action memories,  we have 
$s_{k+1} = \Delta^s_1(k,\delta^1_i(s_k,v_i)) = \Delta^s_1(k,\delta^1_i(s'_k,v_i)) = s'_{k+1}$, as required. 

Next, we show by induction that for all times $k \leq m$ and all agents $i$,
if agent $i$ has performed only $\noop$ to time $m$ in both runs $r$ and $r'$, then 
also $d_k = d'_k$ for all $k \leq m$, 
We have from the fact that the runs $r$ and $r'$ correspond that $r(0) = r'(0)$, so $d_0 = d_0'$. 
Assume  $d_k= d'_k$, for $k<m$. Since the action of agent $i$ in round $k+1$ in both runs $r$ and $r'$ is 
$\noop$, we have that $d_{k+1} = \Delta^s_2(k, \delta_i^2(d_k,\noop)) = 
\Delta^s_2(k, \delta_i^2(d'_k,\noop)) = d'_{k+1}$, as required.  
\end{proof}

\begin{lemma} \label{lem:cornodecsim}
Suppose that information exchange $\exchange$ does not transmit information about actions, 
and the failure model $\failures$ acts independently on 
message and action memory.
Let $P$ and $P'$ be decision protocols with respect to $\exchange$. 
Suppose that $r$ and $\rho$ are runs of $\I_{P, \exchange,\failures}$ and
$r'$ and $\rho'$ are runs of  $\I_{P',  \exchange,\failures}$, 
such that $r$ and $r'$ are corresponding runs and 
$\rho$ and $\rho'$ are corresponding runs. 
Suppose that agent $i$ 
satisfies 
$(r, m) \sim_i (\rho,m)$, and agent $i$ does not decide before time $m$ in 
any of the runs $r,r',\rho,\rho'$. 
Then also $(r', m) \sim_i (\rho',m)$.
\end{lemma} 

\begin{proof}
Since the runs $r$ and $r'$ correspond, as do the runs $\rho$ and $\rho'$, 
and agent $i$ does not decide before time $m$ in any of these runs, 
we have by Lemma~\ref{lem:coreq} that $r_i(m) = r'_i(m)$ and 
$\rho_i(m) = \rho'_i(m)$. From  $(r, m) \sim_i (\rho'm)$ we have that 
$r_i(m) = \rho_i(m)$, and it follows that  $r'_i(m) = \rho'_i(m)$, that is, 
$(r', m) \sim_i (\rho',m)$.
\end{proof}

\newcommand{\dtime}{\mathit{dtime}}
\newcommand{\run}[1]{r^{#1}}

In case the  failure model $\failures$ acts independently on message and action memory,
we say that \emph{state failures in $\failures$ do not perturb the action memory} if for all 
$(\Delta^t,\Delta^r,\Delta^s) \in  \failures$, 
agents $i$,
times $k$ and local states $d \in D_i$, we have $\Delta^s_2(k,d) =d$. 
This means that when 
the information exchange protocol records  decision information, 
the action memory always provides an accurate record of whether 
the agent has made a decision. (This assumption supports satisfaction of the 
Unique Decision property by all agents, even when  state perturbations affect 
the message memory of faulty agents. We note also that in the context of this assumption, 
we can ensure synchrony of the system by including the $\Time_i$ variable as a component of
the action memory.) 

The following result states a sufficient condition for the implementations of the 
knowledge based program to yield an optimum, rather than optimal, implementation. 

\begin{theorem} 
Suppose that information exchange $\exchange$ does not transmit information about actions
and records decision information. Suppose the failures model $\failures$ acts independently on message and action memory, and state failures do not perturb the action memory.
 Let $P$ be an implementation of the knowledge-based program $\kbp(\Phi)$ with respect to $\exchange$ and failure model $\failures$. 
 Then $P$ is an optimum SBA protocol with respect to $\exchange$ and $\failures$. 
\end{theorem} 

\begin{proof} 
For each initial global state $s$, and protocol $P$, there exists a unique run $r$ in $\I_{P,\exchange,\failures}$ 
with $r(0)=s$.  We write $\run{P,s}$ for this run. 
Write $\dtime_i(P,s)$ for the earliest time $m$ at which agent $i$ decides in this run (i.e., for which $P_i(\run{P,s}_i(m)) = \decide_i(v)$ for some value $v$) 
or $\infty$ in case the agent makes no decision.
Since the adversary is encoded into the global state at time 0, we may write $\N(s)$ for the set of faulty agents in this run. 
Note that for two protocols $P$ and $P'$, the runs $\run{P,s}$ and $\run{P',s}$ are corresponding. 

By Proposition~\ref{prop:implementationIsSBA}, since
$P$ is an implementation of the knowledge-based program $\kbp(\Phi)$ with respect to $\exchange$ and failure model $\failures$, we have that $P$ is an SBA protocol with respect to $\exchange$ and $\failures$. 
We derive a contradiction from the assumption that $P$ is not an optimum SBA protocol 
with respect to $\exchange$ and $\failures$. 
Let $P'$ be an SBA protocol with respect to $\exchange$ and $\failures$ such that not $P \leq_{\exchange,\failures} P'$.  
In this case, there exists an initial global state $s$ and an agent $i$ such that $\dtime_i(P',s) < \dtime_i(P,s)$. 
Let $m$ be the least value such that there exists a global state $s$ and an agent $i$ such that $m=\dtime_i(P',s) < \dtime_i(P,s)$. 
Then for all $k<m$, for all initial global states $t$, and all agents $i$, we have that if $k = \dtime_i(P',t)$ then $\dtime_i(P,t) \leq k$.

Since agent $i$ decides at time $m$ in $\run{P',s}$, we have $\I_{P',\exchange,\failures}, (\run{P',s},m) \models \beln_i\cbn \exists v$ for 
some value $v$, by Lemma~\ref{lem:belcbn}. We show that also $\I_{P,\exchange,\failures}, (\run{P,s},m) \models \beln_i\cbn \exists v$. 
Because $P$ implements the knowledge based program $\kbp(\Phi)$, this implies that agent $i$ decides in $\run{P,s}$ 
at time $m$ or earlier, that is, $\dtime_i(P,s) \leq m$, contradicting $m=\dtime_i(P',s) < \dtime_i(P,s)$. 

To show that $\I_{P,\exchange,\failures}, (\run{P,s},m) \models \beln_i\cbn \exists v$, 
consider runs $\run{P,t}$  and  $\run{P,u}$ 
of $\I_{P,\exchange,\failures}$ such that $i \in \N(t)$ and 
$(\run{P,s},m) \sim_i (\run{P,t},m) \approx^*_\N  (\run{P,u},m)$.
We need to show that $\I_{P,\exchange,\failures}, (\run{P,u},m) \models \exists v$. 
For this, we show that 
$(\run{P',s},m) \sim_i (\run{P',t},m) \approx^*_\N  (\run{P',u},m)$.
It then follows from $\I_{P',\exchange,\failures}, (\run{P',s},m) \models \beln_i\cbn \exists v$
that $\I_{P',\exchange,\failures}, (\run{P',u},m) \models \exists v$. 
Because $\exists v$ is a property of the initial state $u$, we then also have that 
$\I_{P,\exchange,\failures}, (\run{P,u},m) \models \exists v$, as required. 

Note first that since  the information exchange records decisions, 
agent $i$ has not decided before time $m$ in $r^{P,s}$, 
and the failure model does not perturb the action memory,  
we have that  the action memory in $r^{P,s}_i(m)$ 
is in $D_i^1$ (recording that $i$ has not yet made a decision). 
Since $(r^{P,s},m) \sim_i (r^{P,t},m)$, we also have that 
the action memory in $r^{P,t}_i(m)$ 
is in $D_i^1$. Since the failure model does not perturb the action memory,
it follows that agent $i$ has not decided before time $m$ in $r^{P,t}$. 

Since agent $i$ decides at time $m$ in $r^{P',s}$, and $P'$ is an SBA protocol, agent $i$ has not decided before time $m$ in $r^{P',s}$.
Thus, by Lemma~\ref{lem:coreq}, we have that $r^{P',s}(m) = r^{P,s}(m)$.

Consider the run $r^{P',t}$.  If agent $i$ decides before time $m$ in $r^{P',t}$, 
then we have a contradiction to the minimality assumption on 
$m$, since agent $i$ does not decide before time $m$ in $r^{P,t}$. Hence agent $i$ has not decided before time $m$ in any of the runs  $r^{P,s}$,  $r^{P',s}$,  $r^{P,t}$,  $r^{P',t}$.  By Lemma~\ref{lem:cornodecsim}, we have that $(r^{P',s},m)\sim_i (r^{P',t},m)$. 

Note that because $P'_i(r^{P',s}_i(m)) = \decide_i(w)$ for some value $w$, and 
$(r^{P',s},m)\sim_i (r^{P',t},m)$,  we also have that $P'_i(r^{P',t}_i(m)) = \decide_i(w)$. Similarly, because 
 $P'_i(r^{P',s}_i(m)) = \noop$, and $(r^{P,s},m)\sim_i (r^{P,t},m)$, we also have $P'_i(r^{P',s}_i(m)) = \noop$.
 As already noted, agent $i$ does not decide before time $m$ in $r^{P,t}$. 
 Hence the pair of corresponding points $(r^{P',t},m)$ and $(r^{P,t},m)$ also provides a witness 
 for the minimality of $m$ as a time where $P'$ decides strictly before $P$. 
 
By an induction on the chain of indistinguishability relations witnessing 
$(\run{P,t},m) \approx^*_\N  (\run{P,u},m)$
that repeats the above arguments, we derive that $(\run{P',t},m) \approx^*_\N  (\run{P',u},m)$, 
as claimed. 
This proves the claim that $(\run{P',s},m) \sim_i (\run{P',t},m) \approx^*_\N  (\run{P',u},m)$.
\end{proof}

\section{A Counter-example} \label{sec:example}

In this section, we present the counter-example promised above, showing that to obtain 
and optimum SBA protocol as an implementation of the knowledge-based program $\kbp(\Phi)$,
it is not enough to assume that the information exchange does not transmit information about 
decisions.  We demonstrate that the implementation $P$ of $\kbp(\Phi)$ 
with respect to an information exchange $\exchange$ and the sending omissions failure model $SO_t$ 
is not always an optimum SBA protocol with respect to $\exchange$ and $SO_t$.  To do so, we provide an SBA protocol $P'$ with respect to $\exchange$ and $SO_t$ such that we do not have $P \leq_{\exchange,SO_t} P'$. We take $\Values = \{0,1\}$ and give the description of $P'$ for an arbitrary number $n$ of agents of which up to 
$t\leq n$  are faulty, but then specialize to $n=4$ and $t=3$ for the counter-example.

The information exchange $\exchange$ is defined as follows. 
The local states $L_i$ of agent $i$ are tuples of the form $\langle \init_i, \known_i, \new_i, \kfaulty_i, \done_i, \Time_i\rangle$, 
where 
\begin{itemize} 
\item $\init_i\in \{0,1\}$ is the agent's initial value, 
\item $\known_i\in \mathcal{P}(\{0,1\})$ is, intuitively, 
the set of values that the agent knows to be the initial value of some agent, 
\item $\new_i\in \mathcal{P}(\{0,1\})$ is, intuitively, 
the set of values that the agent first learned about in the most recent round,
\item $\kfaulty_i \in \mathcal{P}(\Agents)$ is, intuitively, the set of agents that the agent knows to be faulty, 
\item $\done_i\in \{0,1\}$ indicates whether the agent has made a decision, and 
\item $\Time_i$ is the current time. 
\end{itemize}
The initial local states $I_i$ are the states with 
$\known_i=\{\init_i\}$,  $\new_i=\{\init_i\}$,  $\kfaulty = \emptyset$ and $\done_i = \Time_i = 0$. 

Agent $i$'s set of messages $\Msg_i$ contains $\bot$ and messages of the form $\langle n, f\rangle$, 
where $n\subseteq \{0,1\}$ and $f \subseteq \Agents$. Intuitively, $n$ is a set of values that agent $i$ has 
just learned about, and $f$ is a set of agents that agent $i$ knows to be faulty. The message that agent $i$ sends
when it performs action $a$ and has local state $s_i= \langle \init_i, \known_i, \new_i, \kfaulty_i, \done_i,\Time_i\rangle$
is defined as follows: 
\begin{itemize} 
\item If either $\done_i = 1$ or $a = \decide_i(v)$ for some $v\in \{0,1\}$, 
then $\mu_i(s_i,a) = \langle \emptyset, \emptyset\rangle$. Intuitively, if either the agent is in the process of deciding, 
or it has already decided, then it sends a message carrying no information. Note that this is different from sending 
no message, since reception of such a message informs the recipient that agent $i$ did not make a sending omission 
fault in the current round. Effectively, when an agent decides, it stops participating in the protocol, except for 
sending a heartbeat message in each round. 

\item Otherwise $\mu_i(s_i,a) = \langle \new_i, \kfaulty_i \rangle$. That is, if the agent has not yet decided 
and in the current round performs the action $a = \noop$, the agent transmits the set of values it 
has newly learned about, and the set of agents that it knows to be faulty. 
\end{itemize} 

When agent $i$ is in local state $s_i= \langle \init_i, \known_i, \new_i, \kfaulty_i, \done_i,\Time_i\rangle$, 
performs action $a$, and receives vector of messages $(m_1, \ldots,m_n)$  from the other agents, 
agent $i$'s state update $\delta_i(s_i,a,(m_1, \ldots,m_n)) = \langle \init'_i, \known'_i, \new'_i, \kfaulty'_i, \done'_i, \Time'_i \rangle$ is defined as follows. 
Let $J \subseteq \Agents$ be the set of agents from which agent $i$ actually receives a message, 
so that $m_j = \bot$ iff $j \not\in J$. 
For $j \in J$, suppose  $m_j = (n_j,f_j)$. Then 
\begin{itemize} 
\item $\init_i' = \init_i$,  
\item $\known_i' = \known_i \cup \bigcup_{j\in J} n_j$,
\item $\new_i' = \known_i' \setminus \known_i$, 
\item $\kfaulty_i'  = \kfaulty_i \cup (\Agents\setminus J) \cup \bigcup_{j\in J} f_j$,
\item if $a = \decide_i(v)$ for some $v\in \{0,1\}$, then 
 $\done_i =1$, otherwise $\done_i'= \done_i$, and 
\item $\Time_i' = \Time_i + 1$. 
\end{itemize} 
Intuitively, the agent collects in $\known_i'$ 
the values that it has heard about, either previously or as new values transmitted by 
the other agents in the current round. It records an agent $j$ as known to be faulty in  $\kfaulty_i'$
if either it already knew $j$ to be faulty, it does not receive a message from $j$ in the 
current round, or it receives a message saying that $j$ is faulty. 
This completes the description of the information exchange $\exchange$.

It is easily checked from the definition of the message transmission function $\mu_i$ and 
the state update function $\delta_i$ that this information exchange does not transmit information 
about decisions, since the outputs of these functions are independent of the 
value $v$ in case the action is $\decide_i(v)$.  However, the information exchange does transmit information about actions: an agent records that fact that it has decided in its local state, and
the messages it sends may differ once it has decided. 

The protocol $P'$ is defined for agent $i$ on a local state $s_i= \langle \init_i, \known_i, \new_i, \kfaulty_i, \done_i,\Time_i\rangle$, when there may be up to $t$ faulty agents, by 
$P'_i(s_i) = \decide_i(v)$ if $\done_i =0$ and $v$ is the least value in $\known_i$ 
and either $\Time = t+1$ or $\kfaulty_i = \Agents \setminus \{i\}$,  and $P_i(s_i) = \noop$ otherwise.  
That is, an agent decides if it learns that it is the only nonfaulty agent, otherwise it waits to time $t+1$
to make a decision.

Termination is not a requirement of our specification of SBA, 
but we note that the following lemma aids in showing that $P'$ is in fact a terminating SBA protocol. 
When $r$ is a run of a protocol and $m$ a time, we write $\known_i(r,m)$, $\new_i(r,m)$, etc, 
for the components of the local state $r_i(m)$ of agent $i$ in the run $r$. 

\begin{lemma} \label{lem:faultyseq}
Let $r$ be a run of $\I_{P',\exchange,SO_t}$ 
and suppose that for a value $v \in \{0,1\}$ and $k\geq 1$, 
there exist agents $i_0, \ldots , i_{k+1}$ such that $v \in \new_{i_m}(r,m)$
for all $m= 0\ldots k+1$. Then  agents $i_0, \ldots i_{k+1}$ are distinct, and for each $j = 0\ldots k-1$, 
agent $i_j$  has a sending omission fault in round $j+1$. Thus $k\leq t$. 
\end{lemma} 

\begin{proof} 
Note first that we cannot have $i_j = i_{j'}$ for $j\neq j'$, for then the value $v$ is received by some agent for the first time twice on the run, contrary to the update rule for $w_i$ and $\new_i$.  Next, suppose that $i_j$ is not 
faulty for some $j \in \{0, \ldots, k-1\}$. Note $j+2 \leq k+1$. Since $v\in \new_{i_j}(j)$, 
agent $i_j$ transmits a message $(n,f)$ in round $j+1$ with $v \in n$, which agent $i_{j+2}$ receives. 
Thus we have that either $v\in \new_{i_j}(r,m)$ for some $m < j+1$, or $v\in \new_{i_j}(r,j+1)$. Both cases 
contradict the assumption that $v\in \new_{i_j}(r,j+2)$. Thus, we must have that $i_j$ has a sending 
omission fault in round $j+1$. The conclusion that $k \leq t$ follows using the fact that the agents $i_0, \ldots i_{k-1}$ 
are distinct and all are faulty in run~$r$. 
\end{proof} 

Using this lemma, we may show the following. 

\begin{proposition} 
$P'$ is an SBA protocol with respect to $\exchange$ and $SO_t$. 
\end{proposition} 

\begin{proof} 
Unique-Decision holds because an agent performs a $\decide_i(v)$ action only if $\done_i=0$, 
and the variable $\done_i$ captures whether the agent has performed a  $\decide_i(v')$ action
some time in the past.  Validity($\N$) holds because when agent $i$ performs $\decide_i(v)$, 
we have $v \in w_i$, which can be the case only when some agent $j$ had $\init_j =v$, 
by a straightforward induction using the initial condition and update rule for $w_i$. 
 
For Simultaneous-Agreement, suppose that $i\in \N$ performs $\decide_i(v)$ in round $m+1$ of run $r$. 
By definition of $\exchange$, and a straightforward induction, 
the set $\kfaulty_i(r,m)$ contains only faulty agents. 
Hence, in the case where $\kfaulty_i(r,m) = \Agents \setminus \{i\}$, we have that $i$ is the only nonfaulty 
agent in run $r$, and Simultaneous-Agreement holds trivially. Otherwise, suppose that 
$m= t+1$. If any nonfaulty agent $j\neq i$ decided earlier, then $j$ can only have done so because it is the only 
nonfaulty agent, contradicting the assumption that $i$ is nonfaulty. Hence no nonfaulty agent has decided earlier. 
This implies that all nonfaulty agents decide in round $m+1$ also. 
It remains to show that they decide on the same value. We derive a contradiction from the assumption that they do not. In this case, there exist nonfaulty agents $i\neq j$ such that $i$ decides 0 and $j$ decides $1$. 
Since agents decide upon the least value, we must have that $0 \not \in w_j(r,t+1)$. 
Because $i$ is nonfaulty, we must have $0 \in \new_i(r,t+1)$, for if $i$ had a 0 any earlier, 
$j$ would have received 0 by time $t+1$. But for $0 \in \new_i(r,t+1)$, there 
must exist a sequence of agents $i_0,i_1, \ldots i_t, i_{t+1}$ with $i_{t+1} = i$, 
with $0\in \new_{i_j}(r,j)$ for all $j = 0 \ldots t+1$.  By Lemma~\ref{lem:faultyseq}, 
we must have that agents $i_0, \ldots i_{t+1}$ are distinct and $i_0, \ldots i_{t-1}$ are faulty. 
But since there are at most $t$ faulty agents, this means that $i_t$ is nonfaulty and 
$0 \in \new_{i_t}(r,t)$. This implies that $0 \in w_j(r,t+1)$, the required contradiction. 
\end{proof} 

We now argue that for the implementation $P$ of $\kbp(\Phi)$ with respect to $\exchange$ and $SO_t$, 
we do not have that $P \leq_{\exchange,SO_t} P'$. Consider the case of $n=4$ and $t=3$, 
and let $r$ be a run in which the only failures are that agents 1,2, and 3 omit to send 
their message to agent 1 in round 1. Note that the model is defined in such a way that an agent is able to 
detect its own faultiness by seeing that a message it sent to itself was not received. 
Hence, we have $\kfaulty_1(r,1) = \{1,2,3\}$. In case of protocol $P$, this means that 
$\I_{P,\exchange,SO_t},(r,1) \models K_i (i \not \in \N)$,  which implies that 
 $\I_{P,\exchange,SO_t},(r,1) \models \beln_i \cbn \exists v$ for all $v$. 
 According to $P$, therefore, agent $1$ decides in round 2  and sends the message 
$(\emptyset,\emptyset)$ in round 2 (and all subsequent rounds). This means that at time 2,
all other agents $i$ have $\kfaulty_i(r,2) = \emptyset$. The run is indistinguishable 
to the other agents from a run without failures. When $t= n-1$, the earliest possible decision 
time in a run without failures is round $t+1$
(see appendix), but $t=3$, so no nonfaulty agent running $P$
can decide in round 3 in run $r$. 

On the other hand, for protocol $P'$, agent $1$ does not decide in round 2 of the run $r'$
corresponding to $r$, 
since we do not have $\kfaulty_1(r',1) = \Agents \setminus \{1\}$ or $1 = t+1= 4$. 
By the definition of $\exchange$, this means that agent 1 sends a message 
$(w,\{1,2,3\})$ in round 2, and the nonfaulty agents $i$ have 
$\kfaulty_i(r',2) = \{1,2,3\}$.  This means that the nonfaulty agents all decide in round 3 of
the run $r'$. 

We therefore have a run in which the nonfaulty agents decide earlier using $P'$ than they do when using 
the corresponding run of $P$, so it is not the case that  $P \leq_{\exchange,SO_t} P'$.
We remark that this remains the case had we defined $\leq_{\exchange,\failures} $
by comparing decision times of only the nonfaulty agents, rather than all agents. 

Figure~\ref{fig:nodec} shows a key part of the argument for the fact that a decision cannot be made in round three in 
a failure free run. 
The figure depicts a sequence of runs for four agents and indistinguishability relations 
at time 2, from a failure free run (at the top of the diagram) with both 0 and 1 values to a run (at the bottom of the 
diagram) with only 0 values. 
Dashed lines indicate messages that are \emph{not} sent. (We omit messages that are sent in 
order to avoid cluttering the diagram.) $N$ indicates nonfaulty agents. 
This shows that the first run, at time 2, is $\approx^*_\N$ related to the last (also at time 2). 
The first run is indistinguishable to the first agent from a similar run that has three 1 and one 0 value, and a similar sequence then shows that there is also a $\approx_\N$ path to a run with only 1 values. 
It follows that, in a failure free run such as the first, 
we do not have either $\beln_i\cb{\N} \exists 0$ or $\beln_i\cb{\N} \exists 1$. 

\begin{figure} 
\centerline{\includegraphics[width=4in]{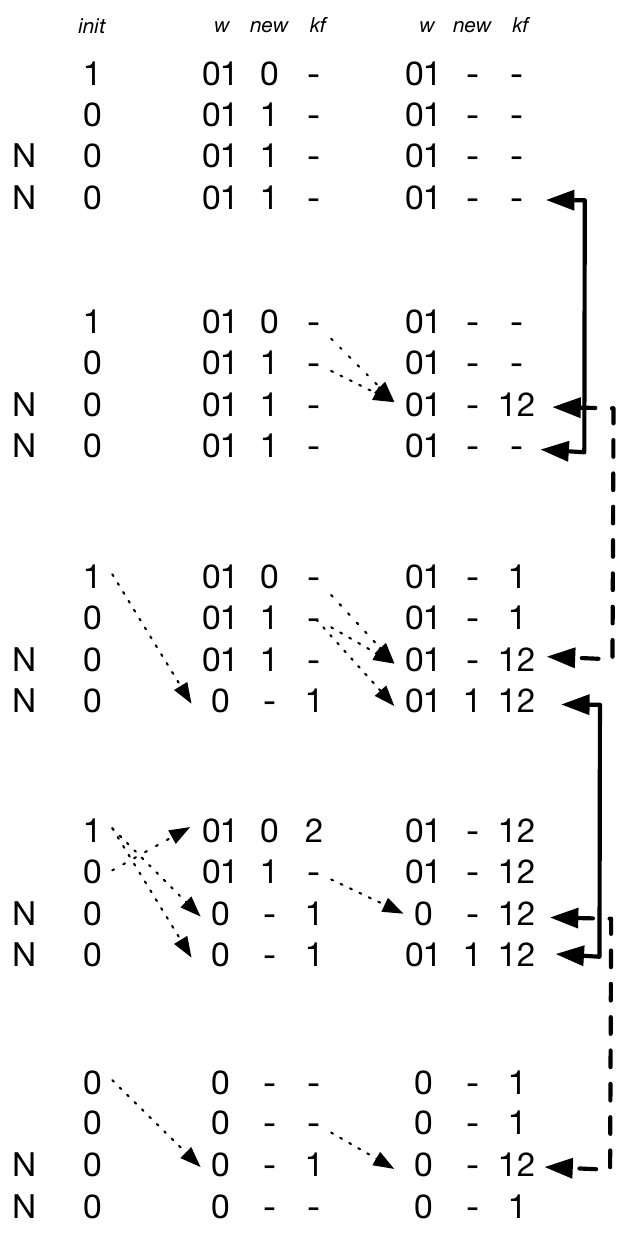}}
\caption{Runs of $\exchange$\label{fig:nodec}} 
\end{figure}

\section{Conclusion} 
\label{sec:concl} 
Our focus has been on \emph{Simultaneous} Byzantine Agreement, in which the 
nonfaulty agents are required to decide at the same time. A number of variants of the 
specification have been studied in the literature on the knowedge based approach to distributed algorithms. 

One dimension of variation is  the behaviour of faulty agents. 
The SBA specification  does not require the faulty agents to make the same
decision as the nonfaulty agents. 
Neiger and Tuttle \cite{NeigerT93} consider 
the \emph{uniform} (also called \emph{consistent}) variant, in which the faulty 
agents, if they decide, must agree with the nonfaulty agents. 
They show that a different 
formulation of common knowledge captures the condition under which a decision can be made, 
which is equivalent to the ``common belief'' condition for the crash and sending 
omissions failure models, but may differ otherwise. In general, the faulty agents cannot 
decide before the nonfaulty agents in this problem, so the example of Section~\ref{sec:example} 
does not apply, and it remains open to understand optimality of Uniform SBA with respect to limited information exchange. 

\roncomment{its probably not too hard to resolve this and add it to the paper! Something odd: 
NT say that CK and CB are equivalent for crash and SO, but why then does the CB kbp allow
fault to decide earlier? Is it just for nonfaulty that they are equivalent? } 

Another dimension of variation is 
simultaneity. In 
\emph{Eventual Byzantine Agreement} (EBA), nonfaulty agents may decide at 
different times. In general, there is not an optimum protocol for this specification, 
but there are optimal protocols. Halpern, Moses and Waarts \cite{HalpernMW01} show that a  more
complex notion called ``continual common knowledge'' is required to 
capture the conditions under which a decision can be made in optimal protocols
for EBA.  Neiger and Bazzi \cite{NeigerB99} show that adding a termination requirement 
to the specification further complicates the required notion of common knowledge. 
We do not presently have a general characterization of optimality with respect to limited information 
exchange for EBA. Alpturer, Halpern and van der Meyden \cite{AHM23} present optimal 
protocols, for full information exchange and for two specific limited information exchanges, but the 
proof of optimality for the latter uses side conditions that do not hold in general. 
In particular, information exchanges involving reports about faults detected, such as 
our example in Section~\ref{sec:example}, do not satisfy these side conditions. 
A satisfactory general characterization of 
optimality for EBA with respect to limited information exchange therefore 
remains open.

We have identified conditions on the information exchange under which the 
knowledge-based program $\kbp(\Phi)$ gives an optimum with respect to a limited information exchange that does 
not transmit information about actions, but also a counter-example 
that shows that  $\kbp(\Phi)$ yields an optimal but not optimum implementation 
when the information exchange transmits information about actions. 
The underlying reason is that the knowledge based program
 forces faulty agents to decide early, and this may diminish the amount of 
information available to the nonfaulty agents. 

Conceivably, another knowledge based program can express the 
optimum implementation, if one exists,  
with respect to an order that compares the decision times of only the nonfaulty agents only.
However,  it would seem that such a program would need agents that discover that they are faulty to determine 
when they decide based on counterfactual reasoning about the consequences, on the decision times
of the nonfaulty agents, of  deciding or deferring a decision. 
This introduces a number of complexities. For one thing, the knowledge-based program would need to 
refer to the future, and a unique implementation of the knowledge based program is then not guaranteed to exist.
Counterfactual reasoning in knowledge based programs also requires a more complex semantic framework, which
has been little studied. (The only relevant work is \cite{HalpernM04}.)  
We therefore leave this question for future work. 

A final issue left for future work is the impact of a change of a number of definitions invovling 
quantifications over the set of all agents. The quantification could instead be made over the set 
of nonfaulty agents. Definitions where this could make sense include the Unique Decision property, 
the order $\leq_{\exchange,\failures}$, and the definition of an implementation of a knowledge-based program. Such alternatives may be of particular interest in Byzantine settings where 
state perturbations do not satisfy the restrictions we have assumed on the failure model 
to obtain a number of our results.

\backmatter 

\bmhead{Acknowledgements}

The Commonwealth of Australia (represented by the Defence Science and Technology
Group) supported this research through a Defence Science Partnerships agreement.

\bibliography{bib}

\end{document}